\documentclass[journal]{IEEEtran}  % Comment this line out if you need a4paper

\IEEEoverridecommandlockouts                              % This command is only needed if 
\usepackage{graphics} % for pdf, bitmapped graphics files
\usepackage{epsfig} % for postscript graphics files
\usepackage{mathptmx} % assumes new font selection scheme installed
\usepackage{times} % assumes new font selection scheme installed
\usepackage{amsmath} % assumes amsmath package installed
\usepackage{amssymb}  % assumes amsmath package installed
\usepackage{psfrag,graphicx,latexsym}
\usepackage{cite}
\usepackage{eurosym}
\usepackage{algorithm}
\usepackage{algpseudocode}
\usepackage{pifont}
\usepackage{gensymb}
\usepackage{color}
\makeatletter
\newlength{\continueindent}
\renewenvironment{algorithmic}[1][0]%
   {%
   \edef\ALG@numberfreq{#1}%
   \def\@currentlabel{\theALG@line}%
   \setcounter{ALG@line}{0}%
   \setcounter{ALG@rem}{0}%
   \let\\\algbreak%
   \expandafter\edef\csname ALG@currentblock@\theALG@nested\endcsname{0}%
   \expandafter\let\csname ALG@currentlifetime@\theALG@nested\endcsname\relax%
   \begin{list}%
      {\ALG@step}%
      {%
      \rightmargin\z@%
      \itemsep\z@ \itemindent\z@ \listparindent2em%
      \partopsep\z@ \parskip\z@ \parsep\z@%
      \labelsep 0.5em \topsep 0.2em%\skip 1.2em 
      \ifthenelse{\equal{#1}{0}}%
         {\labelwidth 0.5em}%
         {\labelwidth 1.2em}%
       \leftmargin\labelwidth \addtolength{\leftmargin}{\labelsep}
      \ALG@tlm\z@%
      }%
      \parshape 2 \leftmargin \linewidth \continueindent \dimexpr\linewidth-\continueindent\relax
   \setcounter{ALG@nested}{0}%
   \ALG@beginalgorithmic%
   }%
   {% end{algorithmic}
   \ALG@closeloops%
   \expandafter\ifnum\csname ALG@currentblock@\theALG@nested\endcsname=0\relax%
   \else%
      \PackageError{algorithmicx}{Some blocks are not closed!!!}{}%
   \fi%
   \ALG@endalgorithmic%
   \end{list}%
   }%
\makeatother

\makeatletter
\def\BState{\State\hskip-\ALG@thistlm}
\makeatother

\allowdisplaybreaks
\newtheorem{theorem}{Theorem}

\newtheorem{proposition}[theorem]{Proposition}

\newtheorem{remark}[theorem]{Remark}

\def\QEDclosed{\mbox{\rule[0pt]{1.3ex}{1.3ex}}} % for a filled box
\def\QED{\QEDclosed} % default to closed

\newcommand{\set}[1]{{\{ #1 \}}}

\def\r{\rho}
\def\f{\varphi}
\DeclareMathOperator{\G}{\Box}
\DeclareMathOperator{\F}{\rotatebox[origin=c]{45}{$\Box$}}

\allowdisplaybreaks

\title{\LARGE \bf
Shrinking Horizon Model Predictive Control with Signal Temporal Logic
Constraints under Stochastic Disturbances
}

\author{Samira S.~Farahani$^{*}$, Rupak Majumdar$^{*}$, Vinayak Prabhu$^{*}$, Sadegh Esmaeil Zadeh Soudjani$^{*}$% <-this % stops a space
\thanks{$^{*}$The authors are with the Max Planck Institute for Software Systems, Germany. 
A limited subset of the results of this paper is accepted for presentation at American Control Conference 2017 \cite{FarMaj:16}. 
        {\tt\small farahani,vinayak,rupak,sadegh@mpi-sws.org}}%
        }

\begin{document}

\maketitle
\thispagestyle{empty}
\pagestyle{empty}

%%%%%%%%%%%%%%%%%%%%%%%%%%%%%%%%%%%%%%%%%%%%%%%%%%%%%%%%%%%%%%%%
\begin{abstract}
We present Shrinking Horizon Model Predictive Control (SHMPC)
for discrete-time linear systems with Signal Temporal Logic
(STL) specification constraints under stochastic disturbances.
The control objective is to maximize an optimization function
under the restriction that a given STL specification is satisfied with
high probability against  stochastic uncertainties.
We formulate a general solution, which does not require precise knowledge
of the probability distributions of the (possibly dependent) stochastic
disturbances; only the bounded support intervals of the density functions and 
moment intervals are used.
For the specific case of disturbances that are independent and normally
distributed, we optimize the controllers further by utilizing  knowledge
of the disturbance probability distributions. We show that in both cases,
the control law can be obtained by solving optimization problems
with linear constraints at each step.
We experimentally demonstrate effectiveness of this approach by
synthesizing a controller for an HVAC system.
\end{abstract}

%
%\keywords{Signal temporal logic; model predictive control; stochastic disturbance}
%

%%%%%%%%%%%%%%%%%%%%%%%%%%%%%%%%%%%%%%%%%%%%%%%%%%%%%%%%%%%%%%%%%%%%%%%%%%%%%%%%
\section{Introduction}
\label{sec_intro}

We consider the control synthesis problem
%\VP{bounded-time control}
for stochastic discrete-time linear 
systems under path constraints that are expressed as temporal logic specifications and are written
in signal temporal logic (STL) \cite{MalNic:04}. 
Our aim is to obtain a controller that robustly satisfies desired temporal properties with high probability 
despite stochastic disturbances, while optimizing additional control objectives.
With focus on temporal properties defined on a finite 
path segment, we use the model predictive control (MPC) scheme 
\cite{BemHee:02-003,DeSvan:99-10,LazHee:06,Mac:02}
with a \emph{shrinking horizon}:
the horizon window is fixed and not shifted at each time step of the controller synthesis 
problem. We start with an initial prediction horizon dependent on the temporal logic
constraints, compute the optimal control sequence for the horizon, apply the
first step, observe the system evolution under the stochastic disturbance, 
and repeat the process (decreasing the prediction horizon by 1)
until the end of the original time horizon.

Our proposed setting requires solving three technical challenges in the MPC framework.

First, in addition to optimizing control and state cost, the derived controller must ensure that 
the system evolution satisfies 
chance constraints arising from the STL specifications. 
Previous choices of control actions can impose temporal constraints on the rest of the path.
We describe an algorithm that updates the temporal constraints based on previous actions.

Second, for some temporal constraints, we may require that the system satisfies the constraints
\emph{robustly}: small changes to the inputs should not invalidate the temporal constraint.
To ensure robust satisfaction, we use a quantitative notion of robustness for STL~\cite{fainekos2009robustness}.
We augment the control objective to
maximize the expected robustness of an STL specification,
in addition to minimizing control and state costs, under chance
constraints.
Unfortunately, the resulting optimization problem is not convex.

As a third contribution, we propose a tractable approximation method for the solution of the 
optimization problem. We conservatively approximate the chance
constraints by linear inequalities.
Second, we provide a tractable procedure to compute
an upper bound for the expected value of the robustness function
under these linear constraints.
%\VP{Is this upper bound a convex function of the control inputs?} 
%which is then optimized under linear constraints. 
%The upper bound is computable as a function of the moments of the disturbance.

Recently receding horizon control with STL constraints has been studied
for a variety of domains \cite{FarRam:15,Raman15}.
In these works, the disturbance is assumed 
to be deterministic but from a bounded polytope, 
and the worst-case MPC optimization problem is solved. 
The control synthesis for deterministic systems with 
probabilistic STL specifications is studied in \cite{SadKap:16} but only a fragment of 
STL is considered in order to obtain a convex optimization problem. 
Also, the receding horizon control has been applied to deterministic 
systems in the presence of perception uncertainty \cite{JhaRam:16}.
Additionally, chance-constrained MPC has been addressed in \cite{SchNik:99} for 
deterministic systems, in which the underlying probability space comes only from the measurement noise.
Application of chance-constrained MPC to optimal control of drinking water networks is studied in \cite{GroOca:14}.
%MPC with STL constraint has been discussed in \cite{raman-cdc14,FarRam:15}.
%The constraints of these works are 
%has been applied to 
%quantile-based  without any temporal logic specifications has 
%been addressed ;  also a quantile-based 
% any temporal logic specifications has been studied for 
%the case that there is only measurement noise and the uncertainty only appears in the constraints.}

In this paper, we assume that the the disturbance signal has an arbitrary probability 
distribution with bounded domain and that we only know the support and the first 
moment interval for each component of the disturbance signal. 
In order to solve the optimization problem more efficiently, 
we transform chance constraints into their equivalent (or approximate) 
linear constraints. To this end, we apply the technique presented by \cite{BouGou:16}, 
to approximate the chance constraints with an upper bound. 
Also, the expected value of the robustness function
can be approximated by the moment intervals of 
the disturbance signal, and can be computed without using numerical integration. 

Furthermore, as an additional case in this study, we show that if the disturbance 
signal is normally distributed and hence, has no bounded support, instead of truncating 
the distribution to obtain a finite interval of support for random variables, we can 
use a different approach, which is based on the quantiles of the normally distributed random variables to 
replace the chance constraints by linear constraints. In this case, we also show that 
the expected value of the robustness function can be replaced by an 
upper bound using the methods presented in \cite{Farvan:16}. 

We empirically demonstrate the effectiveness of our approach by 
synthesizing a controller for a Heating, Ventilation and Air Conditioning (HVAC) 
system.
We compare our approach with open-loop optimal controller synthesis and with robust MPC \cite{Raman15},
and show that our approach can lead to significant
energy savings.

\subsection{Notations}
\label{notation}
%\Sadegh{
We use $\mathbb{R}$ for the set of reals and $\mathbb N := \{0,1,2,\ldots\}$ for the set of non-negative integers.
%and  $\mathbb N_n := \{1,2,\ldots,n\}$ for the finite set of positive integers.
The set $\mathbb B:=\{\top,\bot\}$ indicates logical true and false.
For a vector $v\in\mathbb R^s$, its components are denoted by $v_k$, $k\in\{1,\dots,s\}$.
For a sequence $\{v(t)\in\mathbb R^s,\,\,t\in\mathbb N\}$ and $t_1<t_2$, we define
$\tilde v(t_1:t_2) := [v(t_1),v(t_1+1),\ldots,v(t_2-1)]$, 
%}
In this paper, all random variables are denoted by capital letters and 
the deterministic variables are denoted by small letters.
%\Sadegh{
We also use small letter $y$ to indicate observations of a random vector $Y$.
For a sequence of random vectors $\{Y(t)\in\mathbb R^s,\,\,t\in\mathbb N\}$ and $t_1\le t< t_2$, we define
$\bar Y(t_1:t:t_2) := [y(t_1),y(t_1+1),\ldots,y(t),Y(t+1),\ldots,Y(t_2-1)]$,
which is a matrix containing observations of the random variable up to time $t$ augmented 
with its unobserved values after $t$. 
%}
For a random variable $Y(t)$ denote the support interval by $I_{Y(t)}$ 
and the first moment\footnote{The 
expected value of a random variable $X$ with support $\mathcal{D}$ and 
the cumulative distribution function $F$ is defined as $\mathbb{E}[X] = \int_{\mathcal{D}}xdF(x)$. 
The expectation exists if the integral is well-defined and yields a finite value.}
by $\mathbb{E}[Y(t)]$. 

We consider operations on intervals according to interval arithmetic: for two arbitrary 
intervals $[a,b]$ and $[c,d]$, and constants $\lambda,\gamma\in\mathbb{R}$, 
we have $[a,b]+[c,d] = [a+c,b+d]$ and
$\lambda \cdot [a,b] + \gamma= [\min(\lambda a, \lambda b) + \gamma,\max(\lambda a, \lambda b)+\gamma]$.

\section{Discrete-Time Stochastic Linear Systems}
\label{sec:linsys}
In this paper, we consider systems in discrete-time that can be modeled by linear 
difference equations perturbed by stochastic disturbances. Depending on 
the probability distribution of the disturbance signal, we conduct our study for two cases: 
a) the disturbance signal has an arbitrary probability distribution with a bounded domain 
for which we only know the support and their first moment intervals; and
b) the disturbance signal has a normal distribution.
The first case can be extended to random variables with an 
unbounded support, such as normal or exponential random variables, by truncating 
their distributions.
The specific form of the distribution in the second case enables us to perform a more 
precise analysis using properties of the normal distribution. Note that 
the support of a random variable $X$ with values in $\mathbb{R}^n$ is defined as the set 
$\{x\in\mathbb{R}^N\,|\,\text{Pr}_X[\mathcal{B}(x,r)]>0,\, \forall r>0\}$, where 
$\mathcal{B}(x,r)$ denotes the ball with center at $x$ and radius $r$; alternatively, 
the support can be defined as the smallest closed set $C$ such that $\text{Pr}_X[C]=1$.

%\subsection{Arbitrary probability distributions} %{Random variables with arbitrary probability distribution}
%\label{sec:sys}
Consider a (time-variant) discrete-time stochastic system modeled by the difference equation 
\begin{equation}
\label{eq:discsys}
X(t+1) = A(t) X(t) + B(t) u(t) + W(t),\;\;X(0) = x_0,
\end{equation}
where $X(t)\in\mathbb{R}^{n}$ denotes the state of the system at time instant $t$, 
$u(t)\in\mathbb{R}^{m}$ denotes the control input at time instant $t$, and $W(t)\in\mathbb{R}^s$ 
is a vector of %independent 
random variables, the components of which have %a certain probability distribution; 
either of the above mentioned probability distributions.
The random vector $W(t)$ can be interpreted as the process noise or 
an adversarial disturbance. Matrices $A(\cdot)\in\mathbb{R}^{n\times n}$, 
and $B(\cdot)\in\mathbb{R}^{n\times m}$
%and $B_w(\cdot)\in\mathbb{R}^{n\times s}$
are possibly time-dependent appropriately 
defined system's matrices, and the initial state $X(0)$ is assumed to be known.
We assume that %the components $W_k(t)$ of $W(t)$ 
%have arbitrary probability distributions and that 
$W(0),\ldots,W(t)$ are mutually independent random vectors for all time instants $t$.
%Moreover, we restrict the process noise by the following assumption.
%
%\begin{assumption}
%\label{ass:support}
%For any $k\in\{1,\dots,n\}$ and $t\in\mathbb N$, the process noise $W_k(t)$ 
%has the support $I_{W_k(t)}$ and first moment\footnote{The 
%expected value of a random variable $X$ with support $\mathcal{D}$ and 
%the cumulative distribution function $F$ is defined as $\mathbb{E}[X] = \int_{\mathcal{D}}xdF(x)$. 
%The expectation exists if the integral is well-defined and yields a finite value.}
%$\mathbb{E}[W_k(t)]$. Its support is an interval $I_{W_k(t)}= [a_k,b_k]$ and
%its first moment $\mathbb{E}[W_k(t)]$ belongs to the interval $\mathbb{M}_{{W_k(t)}} = [c_k,d_k]$,
%with known quantities
%$a_k,b_k,c_k,d_k\in\mathbb{R}$.
%\end{assumption}
%
%\begin{remark}
%Based on interval arithmetic, for two arbitrary 
%intervals $[a,b]$ and $[c,d]$, and constants $\lambda,\gamma\in\mathbb{R}$, 
%we have $[a,b]+[c,d] = [a+c,b+d]$ and 
%$\lambda \cdot [a,b] + \gamma= [\min(\lambda a, \lambda b) + \gamma,\max(\lambda a, \lambda b)+\gamma]$.
%\label{rem:interval}
%\end{remark}
%\Sadegh{
Note that, for any $t\in\mathbb N$, the state-space model \eqref{eq:discsys} provides 
the following explicit form for $X(\tau)$, $\tau\ge t$, as a function of $X(t)$, input $u(\cdot)$, 
and the process noise $W(\cdot)$:
\begin{equation}
\label{eq:state}
X(\tau) = \Phi(\tau,t)X(t) + \sum_{k=t}^{\tau-1}\Phi(\tau,k+1) \left(B(k)u(k) + W(k)\right),
\end{equation}
where $\Phi(\cdot,\cdot)$ is the \emph{state transition matrix} of model \eqref{eq:discsys} defined as
\begin{equation*}
\Phi(\tau,t) = \begin{cases}
A(\tau-1)A(\tau-2)\ldots A(t) & \tau>t\ge 0\\
\mathbb I_n & \tau = t\ge 0,
\end{cases}
\end{equation*}
with $\mathbb I_n$ being the identity matrix.
%}
%The explicit form of $X(\cdot)$ in \eqref{eq:state} and the above assumptions on the process noise $W(\cdot)$
%indicate that, for the observed value of $X(t)$ as $x(t)$, $X(\tau)$ is a random variable with the following 
%interval of support and the first moment interval
%\begin{equation}
%I_{X(\tau)} = [\bar{a}_\tau + \bar{C}_\tau,\bar{b}_\tau+\bar{C}_\tau], \;\; \mathbb{M}_{X(\tau)} = [\bar{c}_\tau+\bar{C}_\tau,\bar{d}_\tau+\bar{C}_\tau]
%\label{eq:interval}
%\end{equation}
%where
%$\bar{C}_\tau = \Phi(\tau,t)x(t) + \sum_{k=t}^{\tau-1}\Phi(\tau,k+1) B(k)u(k)$, and $\bar{a}_\tau,\bar{b}_\tau,\bar{c}_\tau$ 
%and $\bar{d}_\tau$ are weighted sum of $a_k,b_k,c_k,d_k,\;k\in\mathbb N$, 
%obtained by using interval arithmetics mentioned in Remark \ref{rem:interval}.

For a fixed positive integer $N$, and a given time instant $t\in\mathbb{N}$, 
let $\tilde{u}(t:N)=[u(t), u(t+1),\ldots,u(N-1)]$ 
(matrix $\tilde{W}(t:N)$ is defined similarly). A \emph{finite stochastic run} of system 
\eqref{eq:discsys} for the time interval $[t:N]$ is defined as $\Xi(t:N) = X(t)X(t+1)\ldots X(N)$, 
which is a finite sequence of states satisfying \eqref{eq:state}. Since each state $X(\tau)$ depends on $X(t),\tilde{u}(t:N)$, and $\tilde{W}(t:N)$, we can rewrite $\Xi(t:N)$ 
in a more elaborative notation as $\Xi_N(X(t),\tilde{u}(t:N),\tilde{W}(t:N))$.
%\Sadegh{
Analogously, we define an \emph{infinite stochastic run} $\Xi = X(t)X(t+1)X(t+2)\ldots$ as an infinite sequence of states.
%}
Stochastic runs will be used in Section \ref{sec:stl} to define the 
system's specifications.
%\subsection{Normal distribution}
%\label{sec:normal}
%In many control applications the disturbance is normally distributed.
%%\Sadegh{
%Such disturbances do not satisfy Assumption \ref{ass:support} but provide more information about the statistics of the runs of the system, thus we address this class separately.
%\begin{assumption}
%\label{ass:normal}
%Random vectors $W(0),W(1),W(2),\ldots$ are mutually independent. 
%For any $t\in\mathbb N$, $W(t)$ is normally distributed with mean 
%$\mathbb E[W(t)] = 0$ and covariance matrix $\Sigma_{W(t)}$, i.e., 
%$W(t)\sim\mathcal{N}(0,\Sigma_{W(t)})$.
%\end{assumption}
%%}
%
%The explicit form of $X(\tau)$ in \eqref{eq:state} and the 
%fact that normal distribution is closed under affine transformations result in normal distribution for
%$X(\tau)$, $\tau\in\mathbb N$,
%under Assumption \ref{ass:normal}.
%Its expected value and covariance matrix with an observed value $x(t)$ of $X(t)$ are
%\begin{align*}
%\mu_\tau &= \Phi(\tau,t)x(t) + \sum_{k=t}^{\tau-1}\Phi(\tau,k+1) B(k)u(k) \text{ and} \\
%\Sigma_\tau &= \sum_{k=t}^{\tau-1}\Phi(\tau,k+1)\Sigma_{W(k)}\Phi(\tau,k+1)^T,
%\end{align*}
%respectively, for $\tau\geq t\geq 0$. 
%
%The stochastic run $\Xi_N(X(t),\tilde{u}(t:N),\tilde{W}(t:N))$ of 
%system \eqref{eq:discsys} for the case of having normally distributed 
%random variables is defined similarly to the one in the previous section. 

\section{Signal Temporal Logic}
\label{sec:stl}
An infinite run of system \eqref{eq:discsys} can be considered as a signal 
$\xi = x(0)x(1)x(2)\dots$, which is a sequence of observed states.
We consider
Signal temporal logic (STL) formulas with bounded-time temporal operators
 defined recursively according to the grammar \cite{MalNic:04}
\begin{equation*}
\varphi ::= \top\mid \pi \mid \neg \varphi \mid\varphi \land \psi  \mid \varphi {\mathcal U}_{[a,b]}\psi
\end{equation*}
where $\top$ is the \emph{true} predicate; $\pi$ is a predicate whose truth value is determined by the sign
of a function, i.e. $\pi = \{\alpha(x)\ge 0\}$ with $\alpha:\mathbb R^n\rightarrow\mathbb R$ being an affine function of state variables;
$\psi$ is an STL formula;
$\neg$ and $\land$ indicate negation and conjunction of formulas;
and ${\mathcal U}_{[a,b]}$ is the \emph{until} operator with $a,b\in\mathbb{R}_{\ge 0}$.
A run $\xi$ satisfies $\varphi$ at time $t$, 
denoted by $(\xi,t) \models \varphi$, if the sequence 
$x(t)x(t+1)\ldots$ satisfies $\varphi$. Accordingly, 
$\xi$ satisfies $\varphi$, if
$(\xi,0) \models \varphi$.

Semantics of STL formulas are defined as follows.
Every run satisfies $\top$.
The run $\xi$ satisfies $\neg \varphi$ if it does not satisfy $\varphi$;
it satisfies $\varphi \land \psi$ if both $\varphi$ and $\psi$ hold.
For a run $\xi = x(0)x(1)x(2)\ldots$ and a predicate  $\pi=\{\alpha(x)\ge 0 \}$, we have
$(\xi,t) \models \pi$
if $\alpha(x(t)) \ge 0$.
Finally,
$(\xi,t) \models \varphi{\mathcal U}_{[a,b]}\psi$ if
$\varphi$ holds at every time step starting from time $t$ before $\psi$ holds, and 
additionally $\psi$ holds at some time instant between
$a+t$ and $b+t$.
Additionally, we derive the other standard operators as follows.
 \emph{Disjunction} $\varphi \lor \psi:=\neg(\neg\varphi\land\neg\psi)$,
the \emph{eventually} operator as $\F_{[a,b]}\varphi := \top{\mathcal U}_{[a,b]} \varphi$,
and the \emph{always} operator as $\G_{[a,b]}\varphi:=\neg \F_{[a,b]}\neg\varphi$.

Thus $(\xi,t) \models \F_{[a,b]} \varphi$ if $\varphi$ holds at some time instant between $a+t$ and $b+t$ and
$(\xi,t) \models \G_{[a,b]} \varphi$ if $\varphi$ holds
at every time instant  between $a+t$ and $b+t$.

\smallskip\noindent\textbf{Formula Horizon.}
The \emph{horizon} of an STL formula $\varphi$ is the smallest $n\in\mathbb N$ such that the following holds for all signals $\xi =x(0)x(1)x(2)\ldots$ and $\xi' =x'(0)x'(1)x'(2)\ldots$:
\begin{align*}
\text{If }x(t+i) = x'(t+i) \text{ for all }i\in\{0,\ldots,n\}\\
\qquad\qquad \text{Then } (\xi,t) \models \varphi  \text{ iff } (\xi',t) \models \varphi.
\end{align*}
Thus, in order to determine  whether  a signal $\xi$ satisfies an STL formula $\varphi$,
we can restrict our attention to the signal prefix $x(0), \ldots, x(\Delta)$ where
$\Delta$ is the horizon of $\varphi$.
This horizon can be upper-approximated by a bound,
denoted by $\text{len}(\varphi)$,
defined to be 
the maximum over the sums of all nested upper bounds on the temporal operators.
Formally, $\text{len}(\varphi)$ is defined recursively as:
\begin{align*}
\varphi &:=  \top \Rightarrow \text{len}(\varphi) = 0, \qquad
\varphi :=  \pi \Rightarrow \text{len}(\varphi) = 0, \\
\varphi &:= \neg\varphi_1 \Rightarrow \text{len}(\varphi) = \text{len}(\varphi_1), \\
\varphi &:= \varphi_1 \land \varphi_2 
 \Rightarrow \text{len}(\varphi) = \max(\text{len}(\varphi_1),\text{len}(\varphi_2)), \\
 \varphi &:= \varphi_1~{\mathcal U}_{[a,b]}~\varphi_2 \Rightarrow \text{len}(\varphi) = b + \max(\text{len}(\varphi_1),\text{len}(\varphi_2)),
\end{align*}
where $\varphi_1,\varphi_2$ and $\psi$ are STL formulas. 
For example, for $\varphi = \square_{[0, 4]} \F_{[3, 6]}\pi$, we have 
$\text{len}(\varphi) = 4+6=10$.
For a given STL formula $\varphi$, it is possible to verify that $\xi\models\varphi$ 
using only the finite run $x(0)x(1)\ldots x(N)$, where $N$ is equal to $\text{len}(\varphi)$.

\smallskip\noindent\textbf{STL Robustness.}
In contrast to the above Boolean semantics, the quantitative semantics of STL \cite{req_mining_hscc2013} assigns to each formula $\f$ a real-valued function $\r^\f$ of signal $\xi$ and $t$ such that
$\r^\f(\xi,t) > 0$ implies $(\xi,t) \models \f$.
Robustness of a formula $\varphi$ with respect to a run $\xi$ at time $t$
 is defined recursively as
\begin{align*}
\r^{\top}(\xi,t) &=  +\infty,\\
\r^{\pi}(\xi,t)&= \alpha(x(t)) \text{ with } \pi = \{\alpha(x)\ge 0\},\\
\r^{\neg \varphi}(\xi,t) &=  -\r^{\varphi}(\xi,t)\\
\r^{\varphi \land \psi}(\xi,t)&= \min (\r^{\varphi}(\xi,t),\r^{\psi}(\xi,t)),\\
\r^{\varphi\hspace{0.05cm}{\mathcal U}_{[a,b]}\psi}(\xi,t)&\!=\!  \max_{i\in [a, b]} \big(\min (\r^{\psi}(\xi,t+i)\!,\! \min_{j \in [0,i)} \r^{\f}(\xi,t+j))\big),
\end{align*}
where $x(t)$ refers to signal $\xi$ at time $t$.
The robustness of the derived formula $\F_{[a,b]}\varphi$ can be worked out to
be $\r^{\F_{[a,b]} \varphi}(\xi,t)= \max_{i\in [a, b]}\r^{\varphi}(\xi,t+i)$; and
similarly for $\G_{[a,b]}\varphi$ as 
$\r^{\G_{[a,b]} \varphi}(\xi,t) = \min_{i\in [a, b]}\r^{\varphi}(\xi,t+i)$.
The robustness of an arbitrary STL formula is computed recursively on the structure of the 
formula according to the above definition, by propagating the values of the functions associated with
each operand using min and max operators.

\smallskip\noindent\textbf{STL Robustness for Stochastic Runs.}
With focus on stochastic runs $\Xi = X(0)X(1)X(2)\dots$ and using the bound of a formula $\varphi$, the finite stochastic run
$\Xi(t:t+N)=X(t)X(1)\ldots X(t+N)$
with $N = \text{len}(\varphi)$ is sufficient to study probabilistic properties of $(\Xi,t)\models\varphi$.
Analogous to the definition of robustness for deterministic run, we can define
\emph{stochastic robustness} $\rho^{\varphi}(\Xi,t)$ of a formula $\varphi$ 
with respect to the run $\Xi$ for times $t$ 
with the stochastic robustness being dependent on
$\Xi(t:t+N)$ and $\varphi$.

Note that a general STL formula $\varphi$ consists of several Boolean and/or temporal operators. 
Due to the system dynamics \eqref{eq:discsys}, the stochastic run $\Xi(t:t+N)$ and $\rho^{\varphi}(\Xi(t:t+N),t)$ are 
both functions of $\tilde{W}(t:t+N)$.
Therefore, $\rho^{\varphi}(\Xi(t:t+N),t)$ is a random variable since affine operators, maximization 
and minimization are measurable functions.

The above definition of robustness implies that, for any formula $\varphi$ and constant 
$\delta\in(0,1)$, a stochastic run $\Xi=X(0)X(1)X(2)\ldots$ 
satisfies $\varphi$ with probability greater than or equal to $1-\delta$, 
if the finite stochastic run $\Xi(0:N)=X(0)X(1)\ldots X(N)$ with $N\ge \text{len}(\varphi)$ satisfies 
$\text{Pr}\left[\r^{\varphi}(\Xi(0:N),0)>0\right]\ge 1-\delta$.

\section{Control Problem Statement}
\label{sec:control}
For system \eqref{eq:discsys} with a given initial state $X(0)=x_0$,
the stochastic disturbance vector $W(t)$ with a given probability distribution, 
STL formulas $\varphi$ and $\psi$, and some constant $N \geq \max(\text{len}(\varphi),\text{len}(\psi))$, the control 
problem can be defined as finding an optimal input sequence $\tilde{u}^\ast(0:N)=[u^\ast(0),\ldots,u^\ast(N-1)]$, 
that minimizes the expected value of a given objective function $J(\tilde{X}(0:N+1),\tilde{u}(0:N))$ subject to constraints 
on states and input variables, where $\tilde{X}(0:N+1)\!=\![X(0),X(1),\ldots,X(N)]$.
%such 
%that $X_1,\ldots,X_N$ are random state variables.
%Due to the probabilistic nature of the cost function, we minimize its expected value.
This optimization 
problem for the time interval $0 \leq t< N$ can be defined as 
\begin{subequations}
\label{eq:open}
\begin{align}
%\begin{array}{ll}
& \min_{\tilde{u}(0:N)}\;\; \mathbb{E}\left[J(\tilde{X}(0:N+1),\tilde{u}(0:N))\right]\quad \mbox{s.t. } \\%[2ex]
& \;\;
%X(t) = A^tx_0 + \sum_{j=0}^{t-1}A^{t-j-1}\left(Bu_j + B_wW_j\right),\;\;1\leq t\leq N \\%[1.5ex]
X(t) = \Phi(t,0)x_0 + \sum_{k=0}^{t-1}\Phi(t,k+1) \left(B(k)u(k) + W(k)\right),\\
%& \hspace{2cm}\;\;\text{for}\;1\leq k \leq N \nonumber\\%[1.5ex]
%&\text{Pr}\left[{\Xi_N}(x(0),u_{\text{old}},w_{\text{realized}},\tilde{u}(t),\tilde{W}(t)) \models \varphi \right] \geq 1-\delta_t  \label{eq:finopt}\\%
&\;\;\text{Pr}\left[{\Xi_N}(x_0,\tilde{u}(0:N),\tilde{W}(0:N)) \models \varphi \right] \geq 1-\delta,  \label{eq:const}\\[1.5ex]
&\;\;\tilde{u}(0:N)\in {U^N}, %[1.5ex]
%&\;\;X(t+k) \sim\mathcal{N}(\mu_k,\Sigma_k) \;\;\text{for}\;1\leq k\leq N \nonumber\\%[1.5ex]
%&W(\ell)\sim\mathcal{N}({\bf 0},\Sigma_{W(\ell)})\;\;\text{for}\;t\leq \ell \leq N \nonumber
%\end{array}
\end{align}
\end{subequations}
where $\mathbb{E}[\cdot]$ denotes the expected value operator and 
%$\tilde{W}(0\!:\!N) = [W_0,\ldots,W_{N-1}]$ is a stochastic vector with 
%mutually independent components, and each $W_t,\;0\leq t \leq N-1$ is a random 
%vector with a given probability distribution;
the closed set $U^N\in\mathbb{R}^{mN}$ specifies the 
constraint set for the input variables.
The chance constraints \eqref{eq:const} state
that for a given $\delta\in(0,1)$, stochastic runs of the system should satisfy $\varphi$ with a probability 
greater than or equal to $1-\delta$.
We consider the following objective function
\begin{equation}
\label{eq:obj_decomp}
J(\tilde{X}(0:N+1),\tilde{u}(0:N)) := J_{\text{robust}}(\tilde{X}(0:N+1)) + J_{\text{in}}(\tilde{u}(0:N)),
\end{equation}
where the first term 
$J_{\text{robust}}(\tilde{X}(0:N+1)) := -\rho^{\psi}(\tilde X(0:N+1),0)$ represents
the negative value of the robustness function on STL formula $\psi$ at time $0$ 
that needs to be minimized;
and the second term $J_{\text{in}}(\tilde{u}(0:N))$ reflects 
%and $J_{\text{state}}(\tilde{u}(t))$, which represent 
the cost on the input variables and can be defined as a linear or a quadratic function. 

Note that optimization problem \eqref{eq:open} is an open-loop optimization 
problem and we cannot incorporate any information related to the process noise 
or the states of the system.
\begin{remark}
The above problem formulation enables us to distinguish the following two cases:
we put the robustness of a formula in the objective function if the system is required 
to be robust with respect to satisfying the formula;
we encode the formula in the probabilistic constraint if only satisfaction of the formula is important.
\end{remark}
%Since in \eqref{eq:open} we have a minimization problem, $J_{\text{robust}}(\tilde{X}(1:N))$ 
%is defined as the negative value of $\rho^{\varphi}(\tilde{X}(1:N),1)$. %$R_{t}^{\varphi}$ 
%Moreover, we consider the expected value of $J(\tilde{X}(1:N),\tilde{u}(1:N))$ %and  $J_{\text{state}}(\tilde{u}(t))$
%in the optimization problem, since by Remark \ref{rem1}, $\rho^{\varphi}(\tilde X(1:N),1)$ %$R_{t}^{\varphi}$ 
%is a random variable. 
%Hence, $J(\tilde{X}(0:N),\tilde{u}(0:N-1)) = \mathbb{E}\right[J_{\text{robust}}(\tilde{X}(0:N))\right] + J_{\text{in}}(\tilde{u}(0:N-1))$. 
%\begin{remark}
%It is possible that the formula $\varphi$ in $J_{\text{robust}}(\tilde{X}(0:N))$ be different 
%than the formula we have in the probabilistic constraints.
%\end{remark}

\subsection{Model Predictive Control}
\label{sec:mpc}
To obtain a more well-behaved control input and to include the information 
about the disturbances, instead of solving the optimization problem 
\eqref{eq:open}, we apply \emph{shrinking horizon model predictive control} (SHMPC),  
which can be summarized as follows: at time step one, we obtain 
a sequence of control inputs with length $N$ (the prediction horizon) to 
optimize the cost function; %based on the predicted behaviour of the system; 
then we only use the first component of the obtained control sequence 
to update the state of the system (or to observe the state in the case of having 
a stochastic disturbance); in the next time step, we fix the first component of the 
control sequence by the first component of the previously calculated 
optimal control sequence and hence, we only optimize for a control sequence 
of length $N-1$. As such, at each time step, the size of the control sequence decreases by 1. 
Note that in this approach, unlike the receding horizon approach, we do not shift 
the horizon at each time step and the end point of the prediction window is fixed. 
MPC allows us to incorporate the new information we obtain about the state 
variables and the disturbance signal, at each time step and hence, to improve the 
control performance comparing with the one of solving the 
open-loop optimization problem \eqref{eq:open}. 

A natural choice for the prediction horizon $N$ in this setting 
with STL constraints is to set it to be greater than or 
equal to the bound of the formula $\varphi$, i.e., $\text{len}(\varphi)$, 
which was defined in the previous section.
% and is defined as the maximum over the sums of 
%all nested upper bounds on the temporal operators. 
This choice 
provides a conservative maximum trajectory length required to 
make a decision about the satisfiability of the formula. 
%For an STL formula $\varphi$, we can define $\text{len}(\varphi)$ as
%\begin{align*}
%\varphi &:= T\mid F \mid \alpha>0 \mid \alpha \leq 0 \Rightarrow \text{len}(\varphi) = 0;\\
%\varphi &:= \varphi_1 \land \varphi_2 \Rightarrow \text{len}(\varphi) = \max(\text{len}(\varphi_1),\text{len}(\varphi_2)); \\
%\varphi &:= \varphi_1 \lor \varphi_2 \Rightarrow \text{len}(\varphi) = \max(\text{len}(\varphi_1),\text{len}(\varphi_2)); \\
%\varphi &:= \G_{[a,b]} ~\psi \Rightarrow \text{len}(\varphi) = b + \text{len}(\psi); \\
% \varphi &:= \varphi_1~{\mathcal U}_{[a,b]}~\varphi_2 \Rightarrow \text{len}(\varphi) = b + \max(\text{len}(\varphi_1),\text{len}(\varphi_2));
%\end{align*}
%where $\varphi_1,\varphi_2$ and $\psi$ are STL formulas. 
%For example, for $\varphi := \square_{[0, 4]} \F_{[3, 6]} \alpha $, we have 
%$\text{len}(\varphi) = 4+6=10$. 

%\Sadegh{
Let $\bar{X}(0:t:N+1)\!=\![x(0),\ldots,x(t),X(t+1),\ldots,X(N)]$
where $x(0),\ldots,x(t)$ are the observed states up to time $t$
and $X(\tau)$ is the random state variable at time $\tau> t$,
%($\bar{W}(0:t:N)$ is defined similarly).
%which the intervals of support and the moment intervals are defined similar to \eqref{eq:interval}, 
also let $\bar{W}(0:t-1:N)\!=\![w(0),\ldots,w(t-1),W(t), W(t+1),\ldots,W(N-1)]$ 
such that $w(0),\ldots,w(t-1)$ are the noise realizations up to time $t-1$
and $W(\tau)$ are random vectors with given probability 
distributions at time $\tau\ge t$.
%and the components of $W(\ell)$, i.e., $W_k(\ell),\;k=1,\dots,n_x$, have 
%the interval of support $I_{W_k(\ell)} = [a_k,b_k]$ and the moment interval 
%$I_{m_{W_k(\ell)}} = [c_k,d_k]$ for $a_k,b_k,c_k,d_k\in\mathbb{R}$. %normally distributed random vectors 
Define $\bar{u}(0:t-1:N)=[u^\ast(0),\ldots,u^\ast(t-1),u(t),\ldots,u(N-1)]$ to be the vector of input variables 
such that $u^\ast(0),\ldots,u^\ast(t-1)$ are the obtained optimal control inputs up to time $t-1$ 
and $u(t),\ldots,u(N-1)$ are the input variables that need to be determined 
at time $t\geq 0$.% and let $\varphi$ denote the STL formula. 
Given STL formula $\varphi$, observations of state variables $x(0),x(1),\ldots,x(t)$, and
designed control inputs $u^\ast(0),\ldots,u^\ast(t-1)$ of system \eqref{eq:discsys},
the \emph{stochastic SHMPC} optimization problem minimizes the expected value of the cost function
%}
%vector of random variables 
%at each time instant $t$,   
%initial state $X(0)\sim\mathcal{N}(\mu_0,I_0)$,
%the initial state $X(0)=x_0$,
 %$X(t)\sim\mathcal{N}(\mu_t,\Sigma_t)$, 
%and the cost function
\begin{align*}
J(\bar{X}(0:t:&N+1),\bar{u}(0:t-1:N)) =\\
&J_{\text{robust}}(\bar{X}(0:t:N+1)) + J_{\text{in}}(\bar{u}(0:t-1:N)),
\end{align*} 
at each time instant $0 \leq t< N$, as follows
%$J(t) = \mathbb{E}\left[J_{\text{robust}}(\tilde{X}(t))\right] + J_{\text{in}}(\tilde{u}(t))$, % at each time step $t$,  
\begin{subequations}
 \label{eq:finopt}
\begin{align}
%\begin{array}{ll}
%\min_{u(t),\ldots,u(N-1)}
&\min_{\tilde u(t:N)}  \;\;\mathbb{E}\left[J(\bar{X}(0:t:N+1),\bar{u}(0:t-1:N)) \right]\;\;\mbox{ s.t. }\label{eq:exp_obj}\\
% \mathbb{E}\left[J_{\text{robust}}(\bar{X}(t)) \right] + J_{\text{in}}(\bar{u}(t))) \nonumber\\%[2ex]
&
%X(t+\ell) = A^{\ell}x(t) + \sum_{j=0}^{\ell-1}A^{\ell-j-1}(Bu(t+j) + W(t+j) )  \\
X(\tau) = \Phi(\tau,t)x(t) + \sum_{k=t}^{\tau-1}\Phi(\tau,k+1) \left(B(k)u(k) + W(k)\right),\nonumber\\
%& \;\;X(k) = A^kx(0) + \sum_{j=0}^{k-1}A^{k-j-1}\left(Bu(j) + W(j)\right)\nonumber\\%[1.5ex]
& \hspace{5.4cm}\;\;\text{for}\;t \leq \tau \leq N\\%[1.5ex]
&\text{Pr}\left[\Xi_N(x_0,\bar{u}(0:t-1:N),\bar{W}(0:t-1:N)) \models \varphi \right] \geq 1-\delta_t  \label{eq:probconst}\\%
%&\text{Pr}\left[\Xi_N(x(t),\bar{u}(t),\bar{W}(t)) \models \varphi \right] \geq 1-\delta_t  \label{eq:finopt}\\[1.5ex]
&\tilde{u}(t:N)\in U^{N-t}, %[1.5ex]
%&\;\;X(t+k) \sim\mathcal{N}(\mu_k,\Sigma_k) \;\;\text{for}\;1\leq k\leq N \nonumber\\%[1.5ex]
%&W(\ell)\sim\mathcal{N}({\bf 0},\Sigma_{W(\ell)})\;\;\text{for}\;t\leq \ell \leq N \nonumber
%\end{array}
\end{align}
\end{subequations}
where the expected value $\mathbb{E}[\cdot]$ in \eqref{eq:exp_obj} is conditioned on observations $\tilde X(0:t+1) = [x(0),\ldots,x(t)]$ and
%where %$\mathbb{E}[\cdot]$ denotes the expected value operator, 
$\delta_t = \delta/N$ for all $t$.
Optimization variables in \eqref{eq:finopt} are the control inputs $\tilde u(t:N) = [u(t),\ldots,u(N-1)]$.
We indicate the argument of minimum by $\tilde u_{opt}(t:N) = [u_{opt}(t),\ldots,u_{opt}(N-1)]$.
 
The complete procedure of obtaining an optimal control sequence using 
SHMPC is presented in Algorithm \ref{alg1}.
Lines 3 to 8 of this algorithm specify the inputs and the parameters used in the algorithm and line 20 specifies 
the output. In line 10, the SHMPC optimization procedure starts for 
each time step $t\in[0,N-1]$. 
In line 11, we solve the optimization problem \eqref{eq:finopt} to 
obtain an optimal control sequence for time instance $t$. In lines 
12 to 16, we check whether the obtained solution satisfies the STL specifications or not; if yes, 
assign the first component of the obtained input sequence to $u^\ast(t)$, 
and if not, the optimization procedure will be terminated.
%We assign the first component of the obtained input sequence to $\tilde{u}^\ast(t)$ 
%in line 19. 
Finally, in line 17, we apply $u^\ast(t)$ to the system \eqref{eq:discsys} and 
% the noise realization $w(t)$
%the first element of the optimal control sequence $\tilde{u}(t:N-1)$, i.e., $u(t)$, 
observe the states at time instant $t$. 
\begin{algorithm}
\caption{}
\label{alg1}
\begin{algorithmic}[1]
\Procedure{Chance-Constrained Stochastic SHMPC}{}
\BState \emph{input}:
\State STL formulas $\varphi$ and $\psi$ and a fixed $\delta\in(0,1)$
\BState \emph{parameters}:
\State $N \geq \max(\text{len}(\varphi),\text{len}(\psi))$
\State probability distributions of the process noise $\{W(t),\;t=0,\ldots,N-1\}$
%\State $N_t$ the shrinking horizon counter and set $N_0=N$ 
%\State $u_{\text{old}} = [];$
%\State $w_{\text{realized}} = [];$
\State Initial state $x_0$
%\State $X(0)\sim\mathcal{N}(\mu_0,I_0)$
\State $\delta_t=\delta/N$ for $t=0,\ldots,N-1$ %and $\varphi_0=\varphi$
%\BState \emph{Initialize}:
%\State $u_{\text{old}} = []$;
%\State $w_{\text{realized}} = []$;
\BState \emph{closed-loop optimization problem}:
\For{$t=0$; $t < N$; $t=t+1$ }
\State Compute $\tilde u_{opt}(t:N) = [u_{opt}(t),\ldots,u_{opt}(N-1)]$ by solving the optimization problem \eqref{eq:finopt} 
\If{ the solution of optimization problem \eqref{eq:finopt} exists} 
 \State Let $u^\ast(t) := u_{opt}(t)$;
 \Else
 \State Return {\it Infeasible Solution} and terminate the optimization procedure;
 \EndIf
%\State $\bar{u}^\ast(t) = u(t)$;
\State Apply $u^\ast(t)$ to the system and observe the value of $X(t+1)$ as $x(t+1)$
%to observe $x(t+1) = Ax(t) + Bu(t) + w(t)$;
%\State $N_{t+1} = N_t - 1$; 
% \State $u_{\text{old}} = [u_{\text{old}};u(t)]$
% \State $w_{\text{realized}} = [w_{\text{realized}};w(t)]$
%\State $\delta_{t+1} = \delta_{\text{new}}$;
%\State $\varphi_{t+1} = \text{UPDATE}(\varphi)$;%\varphi_{\text{new}}$;
\EndFor
\BState \emph{output}:
\State $u^\ast(0:N) = [u^\ast(0),\ldots,u^\ast(N-1)]$
\EndProcedure
\end{algorithmic}
\end{algorithm}

We show in the following theorem that in Algorithm \ref{alg1}, by using the 
shrinking horizon technique, the specific choice of $\delta_t$, and keeping track of the 
control inputs and observed states, the closed-loop system %$x(1)x(2)\dots x(N+1)$ %the obtained optimal control sequence for the entire simulation, i.e., $\tilde{u}^\ast$ 
satisfies the STL specification $\varphi$ with probability greater than or equal to $1-\delta$. %$\text{Pr}\left[ (x(1)x(2)\dots x(N+1))\models \varphi \right]\geq 1-\delta$.
\begin{theorem}
Given
%random vectors $W(t)$ for $t=0,\ldots,N-1$,
a constant $\delta\in(0,1)$ and an STL formula $\varphi$,
if the optimization problems in Algorithm \ref{alg1} are all feasible,
the computed optimal control sequence $\tilde{u}^\ast(0:N) = [u^\ast(0),\ldots,u^\ast(N-1)]$
ensures that the closed-loop satisfy $\varphi$ with probability 
greater than or equal to $1-\delta$.
\end{theorem}
\begin{proof}
%In Algorithm \ref{alg1}, we choose $\delta_t=\delta/N$ for $t=0,\ldot\textsl{}s,N-1$.
Considering the chance constraint \eqref{eq:probconst}
%in optimization problem \eqref{eq:finopt}, 
the probability that a trajectory of the system violates $\varphi$ at time step $t$ 
%, i.e., not being able to find an optimal control sequence that guarantees 
%the satisfaction of the STL formula by the obtained trajectory with probability 
%greater than or equal to $1-\delta_t$, 
is at most $\delta_t$.
This implies that the probability of violating $\varphi$ in the time interval 
%$t=0,\dots,t',\;t'\leq N-1$ is less than or equal to $\sum_{t=0}^{t'}\delta_t$.  
%Hence, the probability of having infeasible solutions in the time steps 
$t=0,\ldots,N-1$ is at most $\sum_{t=0}^{N-1}\delta_t = \sum_{t=0}^{N-1}\delta/N
=\delta$, which proves that the optimal control 
sequence $\tilde{u}^\ast = [u^\ast(0),\ldots,u^\ast(N-1)]$ obtained using Algorithm \ref{alg1} 
results in trajectories that satisfy $\varphi$ with probability greater than or equal 
to $1-\delta.$
\end{proof}

Note that in practice, if at each time step a feasible solution is not found, 
by using the previous control value, i.e., by setting $u^\ast(t) = u^\ast(t-1)$, 
we can give the controller a chance to retry in the next time step after 
observing the next state.
\begin{remark}
%\Sadegh{
The choice of $\delta_t=\delta/N$ is completely arbitrary. In general, the positive constants $\delta_t$
can be picked freely with the condition that $\sum_{t=0}^{N-1}\delta_t=\delta$.
\end{remark}

%\Sadegh{
Computation of the solution of the optimization problem \eqref{eq:finopt} requires addressing two main challenges:
a) the objective function \eqref{eq:exp_obj} depends on the optimization variables $\tilde u(t:N)$ and on 
random variables $\tilde W(t:N)$, thus we have to compute
the expected value as a function of these variables;
and
b) the feasible set of the optimization restricted by the chance constraint \eqref{eq:probconst} is in general difficult to characterize.
We propose approximation methods in Sections \ref{sec:obj} and \ref{sec:const} to respectively address these two challenges. 
%}

\section{Approximating the objective function}
\label{sec:obj}
To solve the optimization problem \eqref{eq:finopt}, one needs to calculate 
the expected value of the objective function. One way to do this is via numerical integration methods \cite{DavRab:84}.
However, numerical integration is in general both cumbersome and time-consuming.
%\Samira{To avoid numerical integration, one can obtain an analytic solution by approximating} 
For example, the method of approximating
the density function of the disturbance with piecewise polynomial functions defined on polyhedral 
sets \cite{BueEng:00,Las:98} suffers from scalability issues on top of the induced approximation error. 
Therefore, in this section, we discuss an efficient method that computes an upper bound for the expected value of the objective 
function and then, minimize this upper bound instead.

We discuss computation of such upper bounds for both cases of process noise with 
arbitrary probability distribution and with normal distribution in Sections 
\ref{subsec:arbit_approx_obj} and \ref{subsec:normal_approx_obj}, respectively.
For this purpose, we first provide a canonical form for the robustness function of a STL formula $\psi$,
%}
%In this section we approximate both the objective function and the probability of the 
%satisfaction of STL formulas with a convex function and linear inequalities, respectively.
%
%The robustness function of an STL formula $\psi$ can be written as a canonical form, 
which is the mix-max or max-min of random variables. This result is inspired 
by \cite{DeSvan:02-004}, in which the authors provide such canonical forms for max-min-plus-scaling 
functions.
\begin{theorem}
\label{thm:J_robust_canonical}
For a given STL formula $\psi$,
%and assuming that $\alpha$ is an affine function 
%of state variables of system \eqref{eq:discsys},
the robustness function $\rho^{\psi}(\Xi(0:N),0)$, and hence the function $J_{\text{robust}}(\bar{X}(0:t:N))$,
can be written into a max-min canonical form
\begin{equation}
\label{eq:J_rob_canonical}
%\rho^{\psi}(\Xi(0:N),0) = 
J_{\text{robust}}(\bar{X}(0:t:N)) = 
 \max_{i\in\{1,\ldots, L\}}\min_{j\in\mathbb \{1,\ldots,m_i\}}\left\{\eta_{ij} + \lambda_{ij}\bar{W}(0:t:N)\right\},
\end{equation}
and into a min-max canonical form
\begin{equation}
\label{eq:J_rob_canonical2}
%\rho^{\psi}(\Xi(0:N),0) = 
J_{\text{robust}}(\bar{X}(0:t:N)) =
 \min_{i\in\set{1,\ldots, K}}\max_{j\in\set{1,\ldots,n_i}}\left\{\zeta_{ij} + \gamma_{ij}\bar{W}(0:t:N)\right\},
\end{equation}
for some integers $K,L,n_1,\ldots,n_K,m_1,\ldots,m_L$,
where $\lambda_{ij}$ and $\gamma_{ij}$ are weighting vectors and 
$\eta_{ij}$ and $\zeta_{ij}$ are affine functions of $\bar{u}(0:t:N)$ and $x_0$.
\end{theorem}
\begin{proof}
 The proof is inductive on the structure of $\psi$ and uses the explicit form of the states in \eqref{eq:state} utilizing the identities
 $-\max(f_1,f_2) = \min(-f_1,-f_2)$ and
 \begin{align*}
  \min&\left(\max(f_1,f_2),\max(g_1,g_2)\right)=\\
 & \max\left( \min(f_1,g_1),\min(f_1,g_2),
 \min(f_2,g_1),\min(f_2,g_2)\right).
 \end{align*}
 for functions $f_1,f_2,g_1$, and $g_2$.
 \end{proof}
 
\subsection{Arbitrary probability distributions with bounded support}
\label{subsec:arbit_approx_obj}

Suppose the elements of the stochastic 
vector $W(t)$, i.e., $W_k(t),\;k\in\{1,\dots,n\}$ have arbitrary probability 
distribution with known bounded support $I_{W_k(t)}= [a_k,b_k]$ and
its first moment $\mathbb{E}[W_k(t)]$ belongs to the interval $\mathbb{M}_{{W_k(t)}} = [c_k,d_k]$,
with known quantities $a_k,b_k,c_k,d_k\in\mathbb{R}$.
%Moreover, each state variable $X(\tau),\,\,\tau\ge t,$ (cf. \eqref{eq:state}) is also a random variable with 
%the interval of support and the moment interval given by \eqref{eq:interval}. 
Under this assumption, the explicit form of $X(\cdot)$ in \eqref{eq:state}
%and the above assumptions on the process noise $W(\cdot)$
implies that, for the observed value of $X(t)$ as $x(t)$, $X(\tau)$ is a random variable with the following 
interval of support and the first moment interval
\begin{equation}
I_{X(\tau)} = [\bar{a}_\tau + \bar{C}_\tau,\bar{b}_\tau+\bar{C}_\tau], \;\; \mathbb{M}_{X(\tau)} = [\bar{c}_\tau+\bar{C}_\tau,\bar{d}_\tau+\bar{C}_\tau]
\label{eq:interval}
\end{equation}
where
$\bar{C}_\tau = \Phi(\tau,t)x(t) + \sum_{k=t}^{\tau-1}\Phi(\tau,k+1) B(k)u(k)$, and $\bar{a}_\tau,\bar{b}_\tau,\bar{c}_\tau$ 
and $\bar{d}_\tau$ are weighted sum of $a_k,b_k,c_k,d_k,\;k\in\mathbb N$, 
obtained by using interval arithmetics mentioned in Section \ref{notation}.

%\subsubsection{Approximating the objective function} 
%Also, recall that at time $0\leq t \leq N-1$,
The objective function in 
%the optimization problem
\eqref{eq:finopt} can be written as 
$\mathbb{E}\left[J_{\text{robust}}(\bar{X}(0:t:N+1)) \right] + J_{\text{in}}(\bar{u}(0:t-1:N)))$ 
and that $J_{\text{robust}}(\bar{X}(0:t:N+1)) = -\rho^{\psi}(\bar X(0:t:N+1),0)$.
Recall that $\bar{X}(0:t:N+1)=[x(0),\ldots,x(t),X(t+1),\ldots,X(N)]$ with observed 
states $x(0),\ldots,x(t)$ of system \eqref{eq:discsys} and random 
states $X(\tau),\,\,\tau>t$. 
The following theorem shows how we can compute an upper bound for 
$\mathbb{E}[J_{\text{robust}}(\bar{X}(0:t:N+1))]$ based on the canonical form of $J_{\text{robust}}$.
\begin{theorem}
\label{thm:J_rob_can_arb}
For a given STL formula $\psi$, 
$\mathbb{E}\left[ J_{\text{robust}}(\bar{X}(0:t:N+1))\right]$ can be upper bounded by 
$$\max_{i\in\{1,\dots,L\}}\min_{j\in\{1,\dots,m_i\}}(\hat{d}_{ij}+\eta_{ij}),$$
where the constants $\eta_{ij}$, $i\in\{1,\dots,L\}, j\in\{1,\dots,m_i\}$, are affine functions of $\bar{u}(0:t-1:N)$ and $x(0)$,
and $\hat{d}_{ij}$
are weighted sum of $w(0),\ldots,w(t-1)$ and $c_k,d_k$ for $k=t,\ldots,N-1$.
\end{theorem}
\begin{proof}
With focus on the canonical form \eqref{eq:J_rob_canonical},
let $Y_{ij} = \eta_{ij} + \lambda_{ij}\bar{W}(0:t:N)$.
% where $\lambda_{ij}$ is 
%a weighting vector and $\eta_{ij}$ is an affine function of $\bar{u}(0:t-1:N)$ and $x(0)$
%for any $j\in\{1,\dots,m_i\},\;i\in\mathbb \{1,\dots,L\}$. 
%Since $\rho^{\psi}(\bar X(0:t:N+1),0)$ is defined as arbitrary nesting of 
%maximum and/or minimum of some variables, we can rewrite 
%$J_{\text{robust}}(\bar{X}(0:t:N+1))$, using Theorem \ref{thm:J_robust_canonical}, 
%in its canonical form as
%\begin{equation}
%\label{eq:J_rob_can_arb}
%\max_{i\in\{1,\dots,L\}}\min_{j\in\{1,\dots,m_i\}}Y_{ij}.
%\end{equation}
Considering the support and 
moment interval of the components of $W(\tau),\tau=t,\ldots,N-1$, each 
random variable $Y_{ij}$ has the following support 
and moment interval (similar to \eqref{eq:interval})
\begin{equation}
I_{Y_{ij}} = [\hat{a}_{ij} + \eta_{ij},\hat{b}_{ij}+\eta_{ij}], \;\; \mathbb{M}_{Y_{ij}}= [\hat{c}_{ij}+\eta_{ij},\hat{d}_{ij}+\eta_{ij}]
\label{eq:interval2}
\end{equation}
where the constants $\hat{a}_{ij},\hat{b}_{ij},\hat{c}_{ij},\hat{d}_{ij}$, $i\in\{1,\dots,L\}, j\in\{1,\dots,m_i\}$,
are weighted sum of $w(0),\ldots,w(t-1)$ and $a_k,b_k,c_k,d_k$ for $k=t,\ldots,N-1$, using interval arithmetic
(cf. Section \ref{notation}).
%The constants $\hat{C}_{ij}$ are \Sadegh{affine functions of 
%input variables and initial state $x(0)$}, as explained in Section \ref{sec:sys}. 
Accordingly, $J_{\text{robust}}$ 
is a random variable with the following support and moment intervals, 
\begin{align}
&I_{J_{\text{robust}}} = \nonumber\\
% \rho^{\varphi}(\bar X(0:t:N+1),0)
&\hspace{0.4cm} [\max_{i\in\{1,\dots,L\}}\min_{j\in\{1,\dots,m_i\}}(\hat{a}_{ij} \!+\! \eta_{ij}),\max_{i\in\{1,\dots,L\}}\min_{j\in\{1,\dots,m_i\}}(\hat{b}_{ij}\!+\!\eta_{ij}]\nonumber\\
&\mathbb{M}_{J_{\text{robust}}} = \nonumber \\
&\hspace{0.4cm} [\max_{i\in\{1,\dots,L\}}\min_{j\in\{1,\dots,m_i\}}(\hat{c}_{ij}\!+\!\eta_{ij}),\max_{i\in\{1,\dots,L\}}\min_{j\in\{1,\dots,m_i\}}(\hat{d}_{ij}\!+\!\eta_{ij})].
\label{eq:J_rob_interval}
\end{align}
Hence, as we are minimizing the cost function in \eqref{eq:finopt}, we can 
utilize the upper bound
$\max_{i\in\{1,\dots,L\}}\min_{j\in\{1,\dots,m_i\}}(\hat{d}_{ij}+\eta_{ij})$
for
$\mathbb{E}\left[ J_{\text{robust}}(\bar{X}(0:t:N+1))\right]$.
%swap the max and min in \eqref{eq:J_rob_interval} as well.
\end{proof}

Note that the approximation methodology of Theorem \ref{thm:J_rob_can_arb} is 
applicable also to the min-max canonical form \eqref{eq:J_rob_canonical2}.

By replacing the expected objective function by its upper bound given in Theorem \ref{thm:J_rob_can_arb}, and by 
replacing the probabilistic constraints by their equivalent linear approximation (as 
is discussed in Section \ref{sec:const}), the optimization problem \eqref{eq:finopt} 
can be then recast as a mixed integer linear programming (MILP) problem, which can be solved 
using the available MILP solvers \cite{AtaSav:05,JLiTRa:05}.

%$I_{\rho^{\varphi}(\bar X(0:N),0)}$ %in the expressions for $I_{\rho^{\varphi}(\bar X(0:N),0)}$ 
%and $\mathbb{M}_{\rho^{\varphi}(\bar X(0:N),0)}$ as well.

\subsection{Normal distribution}
\label{subsec:normal_approx_obj}
The upper bound on the objective function provided in the previous section does not apply to process noises with unbounded support,
but knowing the distribution of the process noise 
%In many control applications the disturbance is normally distributed.
%Such disturbances do not have a bounded support, as in the previous case, but
provides more 
information about the statistics of the runs of the system.
In this section we address process noises with normal distribution separately due the their wide use in engineering applications.

Suppose that for any $t\in\mathbb N$, $W(t)$ is normally distributed with mean 
$\mathbb E[W(t)] = 0$ and covariance matrix $\Sigma_{W(t)}$, i.e., 
$W(t)\sim\mathcal{N}(0,\Sigma_{W(t)})$.
The explicit form of $X(\tau)$ in \eqref{eq:state} and the fact that normal distribution 
is closed under affine transformations result in normal distribution for $X(\tau)$, $\tau\in\mathbb N$. 
%under Assumption \ref{ass:normal}. 
Its expected value and covariance matrix with an 
observed value $x(t)$ of $X(t)$ are
\begin{align*}
\mu_\tau &= \Phi(\tau,t)x(t) + \sum_{k=t}^{\tau-1}\Phi(\tau,k+1) B(k)u(k) \text{ and} \\
\Sigma_\tau &= \sum_{k=t}^{\tau-1}\Phi(\tau,k+1)\Sigma_{W(k)}\Phi(\tau,k+1)^T,
\end{align*}
respectively, for $\tau\geq t\geq 0$. 

In this section we use the canonical representation of $J_{\text{robust}}(\bar{X}(0:t:N+1))$ 
in Theorem \ref{thm:J_robust_canonical},
%to show how we can compute an upper bound for $\mathbb E\left[J_{\text{robust}}\right]$ based on higher order moments of $W(t)\sim\mathcal{N}(0,\Sigma_{W(t)})$.
%
%Theorem \ref{thm:J_robust_canonical}
which states that $J_{\text{robust}}$ (for
fixed $\bar{u}(0:t:N)$ and $x_0$)
can be written in either of the forms
\begin{equation}
\label{eq:Yij}
\max_{i\in\set{1,\ldots, L}}\min_{j\in\set{1,\ldots, m_i}}Y_{ij}\quad\text{ or }\quad \min_{i\in\set{1,\ldots, K}}\max_{j\in\set{1,\ldots,{n_i}}}Y_{ij}
\end{equation}
with $Y_{ij} = \eta_{ij} + \lambda_{ij}\bar{W}(0:t-1:N)$ being affine functions of the process noise, 
thus normally distributed random variables (similar to $X(\tau)$ explained above). 
With focus on these canonical representations for $J_{\text{robust}}$ we employ 
Proposition \ref{prop:prop1} to show how to approximate 
$\mathbb E\left[J_{\text{robust}}\right]$ using higher order moments of $W(t)\sim\mathcal{N}(0,\Sigma)$.
%
%and use it in Corollaries \ref{cor:normal_maxmin} 
%and \ref{cor:normal_minmax} to show how to compute the required upper bound on 
%$\mathbb E\left[J_{\text{robust}}\right]$ based on higher order moments of $W(t)\sim\mathcal{N}(0,\Sigma)$.
This proposition, also used in \cite{Farvan:16},
follows due to the relation between the infinity norm and the $p$-norm of a vector and
Jensen's inequality.
% The second inequality is the application of Jensen's inequality to the concave function $z^{1/p}$.
% This proposition is already used in \cite{Farvan:16}.
\begin{proposition}
\label{prop:prop1}
Consider random variables $Z_i$ for $i\in\{1,\dots,s\}$ and let $p$ be an even integer.
Then
\begin{align*}
\mathbb{E}\left[ \mathrm{max}(Z_1,\ldots,Z_s) \right]&\leq \mathbb{E}\left[\mathrm{max}(|Z_1|,\ldots,|Z_s|) \right]  \nonumber\\
&\leq\mathbb{E}\left[ ((Z_1)^p+\ldots+(Z_s)^p)^{1/p} \right]  \nonumber \\
&\leq \left( \sum_{i=1}^s \mathbb{E}\left[ (Z_i)^p  \right] \right)^{1/p}.
\end{align*}
\end{proposition}
%
%Using Proposition \ref{prop:prop1}, we prove the following corollaries that 
%give an upper bound on the expected values of quantities in \eqref{eq:Yij} as a function of moments of $Y_{ij}$.
Founded on Proposition \ref{prop:prop1}, next theorem shows how we can upper bound $\mathbb E\left[J_{\text{robust}}\right]$ using the higher order moments of $Y_{ij}$.
\begin{theorem}
Considering the canonical forms in \eqref{eq:Yij} for $J_{\text{robust}}$ as a function of
random variables $Y_{ij}$, $\mathbb E\left[J_{\text{robust}}\right]$ 
can be upper bounded by 
\begin{align}
\mathbb{E}\left[\max_{i\in\set{1,\ldots, L}}\min_{j\in\set{1,\ldots, {m_i}}}Y_{ij}\right] \!\le \! \left( \sum_{i=1}^{L}\sum_{j=1}^{m_i}\mathbb{E}[Y_{ij}^p] \right)^{1/p}, \label{eq:app1}\\
\mathbb{E}\left[\min_{i\in\set{1,\ldots, K}}\max_{j\in\set{1,\ldots, {n_i}}}Y_{ij}\right] \!\le \!  \min_{i\in\set{1,\ldots, K}}\left( \sum_{j=1}^{n_i}\mathbb{E}\left[ Y_{ij}^p\right] \right)^{1/p}.
\label{eq:app2}
\end{align}
%depending on the canonical form of $J_{\text{robust}}$.
\end{theorem}
\begin{proof}
For random variables $Y_{ij},\;i\in\set{1,\ldots, L},\;j\in\set{1,\ldots, {m_i}}$, and for a 
positive even integer $p$, the following inequality holds,
\begin{align*}
\mathbb{E}\left[\max_{i\in\set{1,\ldots,L}}\min_{j\in\set{1,\ldots, {m_i}}}Y_{ij}\right] &\stackrel{(i)}{\leq}  \left(\sum_{i=1}^{L}\mathbb{E}\left[\min_{j\in\set{1,\ldots, {m_i}}}Y_{ij}\right]^p\right)^{1/p} \nonumber\\
& \stackrel{(ii)}{=} \left(\sum_{i=1}^{L}\mathbb{E}\left[-\max_{j\in\set{1,\ldots, {m_i}}}-Y_{ij}\right]^p\right)^{1/p}\nonumber\\
&\stackrel{(iii)}{\leq} \left( \sum_{i=1}^{L}\sum_{j=1}^{m_i}\mathbb{E}[Y_{ij}^p] \right)^{1/p}, 
\end{align*}
where in $(i)$ we used the upper bound obtained 
in Proposition \ref{prop:prop1}; in $(ii)$ we used the fact that 
$\min_{k\in\set{1,\ldots, r}}(\alpha_k) = -\max_{k\in\set{1,\ldots, r}}(-\alpha_k)$;
In $(iii)$ we use again the inequality in Proposition \ref{prop:prop1}. 
Moreover, for $i\in\set{1,\ldots, K},\;j\in\set{1,\ldots, n_i}$, the following inequality holds,
\begin{align*}
\mathbb{E}\left[\min_{i\in\set{1,\ldots, K}}\max_{j\in\set{1,\ldots, n_i}}Y_{ij}\right]
&\stackrel{(i)}{\leq} \min_{i\in\set{1,\ldots, K}}\mathbb{E}\left[ \max_{j\in\set{1,\ldots, {n_i}}}Y_{ij} \right]\\
&\stackrel{(ii)}{\leq} \min_{i\in\set{1,\ldots, K}}\left( \sum_{j=1}^{n_i}\mathbb{E}\left[ Y_{ij}^p\right] \right)^{1/p}\hspace{-0.1in},
\end{align*}
where we apply Jensen's inequality to the concave 
function $\min(\cdot)$ to get $(i)$. The inequality of Proposition 
\ref{prop:prop1} gives $(ii)$. 
\end{proof}

Note that random variables $Y_{ij}$ are normally distributed in both \eqref{eq:app1} and \eqref{eq:app2}.
Higher order moments of normally distributed random variables can be computed analytically in a closed form as a function of the first two moments, i.e., using its mean and variance.
More specifically, for a normally distributed random variable $Z$ with mean $\mu$ and variance $\sigma^2$, 
the $p$-th moment has a closed form as
\begin{equation}
\label{eq4:nor} 
\mathbb{E} \left[Z^{p} \right] = \sigma^{p} i^{-p}H_{p}(i\mu/\sigma)
\end{equation}
where $i$ is the imaginary unit and 
\begin{equation}
 H_{p}(z) = p!\sum_{l=0}^{p/2}\frac{(-1)^l z^{p-2l}}{2^l l!(p-2l)!}
\label{eq4:herm}
\end{equation}
is the $p$-th Hermite polynomial \cite[Chapter 22 and 26]{AbrSte:64}.
We use \eqref{eq4:nor} to compute higher order moments of normal random variables with $p$ being even integers.
Note that the right-hand side of \eqref{eq4:nor} is in fact real because $H_{p}(z)$ contains only even powers of $z$ when $p$ is even.

% \VP{The variable $Z$ in Eq~\eqref{eq4:nor}  is
% an affine sum of $u$ and $W$ (Theorem~\ref{thm:J_robust_canonical}).
% So the mean will depend on $u$. 
% Which means that the $z$ in $H_p(z)$ will depend on $u$.

In the next section we discuss how to cope with the second challenge of 
characterizing the feasible set of the optimization restricted by the chance 
constraint \eqref{eq:probconst}.

\section{Under Approximation of Chance Constraints}
\label{sec:const}

In this section, we discuss methods for computing conservative lower approximations of 
the chance constraints in \eqref{eq:probconst} as linear constraints. 
%the chance constraints in optimization problem \eqref{eq:finopt}.
For the sake of compact notation, we indicate the stochastic run $\Xi(0:N)=X(0)X(1)\ldots X(N)$ only by $\Xi_N$ 
without declaring its dependency on the state, input, and disturbance variables.
Recall the chance constraint \eqref{eq:probconst} as
$\text{Pr}\left[\!(\Xi_N,t) \!\models\! \varphi \! \right]\ge 1-\delta_t.$
In order to transform this constraint to linear inequalities,
we first show in the following theorem, that this constraint can be transformed into similar inequalities but $\varphi$ being an atomic predicate.
%for any formula 
%$\varphi$ and a constant $\vartheta\in(0,1)$, constraints of the forms
%\begin{equation}
%\label{eq:inductive_reduction}
%\text{Pr}\left[\!(\Xi_N,t) \!\models\! \varphi \! \right]\ge\vartheta\text{ and }
%\text{Pr}\left[\!(\Xi_N,t) \!\models\! \varphi \! \right]\le\vartheta
%\end{equation}
%can be transformed into constraints similar to \eqref{eq:inductive_reduction} with 
%$\varphi$ being an atomic predicate.
Then in Sections \ref{sec:probtrans} and \ref{sec:prob}, 
we discuss how to transform the resulting constraints with atomic predicates into linear inequalities for the cases of arbitrary random variables with known bounded 
support and moment interval and of normally distributed random variables.
%can be conservatively transformed into constraints similar to \eqref{eq:inductive_reduction} with sub-formulas of $\varphi$.
%Then we show how to transform constraints \eqref{eq:inductive_reduction}, with $\varphi$ being an atomic predicate, into linear inequalities.
%The first phase is done inductively on the structure of the formula $\varphi$ as follows.
%The second phase is discussed in Section \ref{sec:probtrans} for arbitrary random variables with known bounded support and moment interval
%and in Section \ref{sec:prob} for normally distributed random variables.
\begin{theorem}
for any formula 
$\varphi$ and a constant $\vartheta\in(0,1)$, constraints of the forms
\begin{equation}
\label{eq:inductive_reduction}
\text{Pr}\left[\!(\Xi_N,t) \!\models\! \varphi \! \right]\ge\vartheta\text{ and }
\text{Pr}\left[\!(\Xi_N,t) \!\models\! \varphi \! \right]\le\vartheta
\end{equation}
can be transformed into similar constraints with $\varphi$ being an atomic predicate using the structure of $\varphi$.
\end{theorem}

\begin{proof}
%\Samira{The inequality constraints \eqref{eq:inductive_reduction} can be conservatively 
%transformed into constraints similar to \eqref{eq:inductive_reduction} with sub-formulas of $\varphi$, 
%and ultimately with $\varphi$ being an atomic predicate. This can be done inductively 
%on the structure of the formula $\varphi$ as follows: }
The proof is inductive on the structure of the formula $\varphi$ as discussed in the following three cases.
\medskip

\noindent
{\bf Case I: $\varphi = \neg \varphi_1$}
we have the following equivalences
 \begin{align*}
\text{Pr}\left[(\Xi_N,t) \models \neg\varphi_1  \right]\,\,\ge\,\,\vartheta
\Leftrightarrow &  \text{Pr}\left[\!(\Xi_N,t) \nvDash \varphi_1  \right]\,\,\ge\,\,\vartheta \nonumber \\
 \Leftrightarrow & \text{Pr}\left[(\Xi_N,t) \models \varphi_1  \right]\,\,\le\,\,1-\vartheta,\\ %\label{case21}\\
\text{Pr}\left[\!(\Xi_N,t) \models \neg\varphi_1  \right]\,\,\le\,\,\vartheta
 \Leftrightarrow &  \text{Pr}\left[\!(\Xi_N,t) \nvDash \varphi_1  \right]\,\,\le\,\,\vartheta \nonumber \\
 \Leftrightarrow & \text{Pr}\left[(\Xi_N,t) \models \varphi_1  \right]\,\,\ge\,\,1-\vartheta. %\label{case22}
 \end{align*}
%In this case, $\rho^{\varphi}(\Xi_N,t) = \min_{i=a,\ldots,b}\rho^{\psi}(\Xi_N,t+i)$. 
%Hence, for a given $\delta_t\in(0,1)$, and by using the fact that $\text{Pr}[A\wedge B]\geq 1-\delta_t \Leftrightarrow
%\text{Pr}(\neg A \vee \neg B) \leq \delta_t$ and $\text{Pr}( A \vee B) \leq \text{Pr}[A] + \text{Pr}[B]$, 
%we have
%\begin{align}
%&\hspace{-1.5cm}\text{Pr}\left[ \rho^{\varphi}(\Xi_N,t) > 0 \right] \geq 1-\delta_t \nonumber \\
%\Leftrightarrow & \text{Pr}\left[ \left(\min_{i=a,\ldots,b}\rho^{\psi}(\Xi_N,i) \right) > 0 \right] \geq 1-\delta_t \nonumber \\
%\Leftrightarrow & \text{Pr}\left[ \bigwedge_{i=a}^{b}(\rho^{\psi}(\Xi_N,i) > 0) \right] \geq 1-\delta_t \nonumber \\
%\Leftrightarrow& \text{Pr}\left[\bigvee_{i=a}^{b} (\rho^{\psi}(\Xi_N,i) \leq 0) \right] \leq \delta_t \nonumber\\
%\Leftarrow& \sum_{i=a}^{b} \text{Pr}\left[\rho^{\psi}(\Xi_N,i) \leq 0 \right] \leq \delta_t.
%\label{eq:prob1}
%\end{align} 
%In order to recursively solve constraint \eqref{eq:prob1}, we assume that the 
%probability of each term in the summation is less than 
%or equal to $\delta_t/n$, i.e., $\text{Pr}\left[\rho^{\psi}(\Xi_N,i)\leq 0 \right]\leq \delta_t/n$ 
%for $i\in[a,b]$ with $n=b-a+1$. \\
{\bf Case II: $\varphi = \varphi_1\wedge\varphi_2$}
we obtain the following inequalities by using the fact that for possibly joint events $\mathcal{A}$ and $\mathcal B$, 
it holds that $\text{Pr}[\mathcal A\wedge \mathcal B]\geq\vartheta \Leftrightarrow
\text{Pr}(\neg \mathcal A \vee \neg \mathcal B) \leq 1-\vartheta$ and $\text{Pr}( \mathcal A \vee \mathcal B) \leq \text{Pr}[\mathcal A] + \text{Pr}[\mathcal B]$.
\begin{align}
&\hspace{-0.5cm}\text{Pr}\left[ (\Xi_N,t) \models\varphi_1\wedge\varphi_2 \right] \geq \vartheta \nonumber \\
\Leftrightarrow & \text{Pr}\left[ (\Xi_N,t) \models\varphi_1 \wedge (\Xi_N,t) \models\varphi_2 \right] \geq \vartheta \nonumber \\
\Leftrightarrow & \text{Pr}\left[ (\Xi_N,t) \nvDash\varphi_1 \vee (\Xi_N,t) \nvDash\varphi_2 \right] \leq 1-\vartheta \nonumber \\
\Leftarrow & \text{Pr}\left[ (\Xi_N,t) \nvDash\varphi_1\right] + \text{Pr}\left[  (\Xi_N,t) \nvDash\varphi_2 \right] \leq 1-\vartheta \nonumber \\
\Leftarrow & \text{Pr}\left[ (\Xi_N,t) \nvDash\varphi_i\right] \leq \frac{1-\vartheta}{2} \;\;\;i=1,2.
\label{eq:case31} 
\end{align}
Note that in the last line of \eqref{eq:case31}, we assume that the probability of the two events are upper bounded by the same value, i.e., $(1-\vartheta)/2$.
However, this can be replaced by any two other probabilities $\delta_1$ and $\delta_2$ such that $\delta_1+\delta_2=1-\vartheta$.
Now consider the second possibility:
\begin{align}
&\hspace{-0.3cm}\text{Pr}\left[ (\Xi_N,t) \models\varphi_1\wedge\varphi_2 \right] \leq \vartheta \nonumber \\
\Leftrightarrow &\text{Pr}\left[ (\Xi_N,t) \models\neg\varphi_1\vee\neg\varphi_2 \right] \geq 1-\vartheta \nonumber \\
\Leftrightarrow &\text{Pr}\left[ (\Xi_N,t) \models\neg\varphi_1\vee(\varphi_1\wedge\neg\varphi_2) \right] \geq 1-\vartheta \nonumber \\
\Leftrightarrow &  \text{Pr}\left[ (\Xi_N,t)\models\neg\varphi_1 \right] \!+\! \text{Pr}\left[ (\Xi_N,t) \models\varphi_1\wedge\neg\varphi_2\right] \geq 1-\vartheta, \label{eq:case32}
%\Leftrightarrow & \text{Pr}\left[ (\Xi_N,t) \models\varphi_1 \wedge (\Xi_N,t) \models\varphi_2 \right] \leq \vartheta \nonumber \\
%\Leftrightarrow & \text{Pr}\left[ (\Xi_N,t) \nvDash\varphi_1 \vee (\Xi_N,t) \nvDash\varphi_2 \right] \geq 1-\vartheta \nonumber \\
%\Leftrightarrow &  \text{Pr}\left[ (\Xi_N,t) \nvDash\varphi_1 \vee ((\Xi_N,t) \models\varphi_1\wedge (\Xi_N,t) \nvDash\varphi_2) \right] \geq 1-\vartheta \nonumber \\
%\Leftrightarrow &  \text{Pr}\left[ (\Xi_N,t) \nvDash\varphi_1 \right] \!+\! \text{Pr}\left[ ((\Xi_N,t) \models\varphi_1\wedge (\Xi_N,t) \nvDash\varphi_2) \right] \geq 1-\vartheta,
%\label{eq:case32}
\end{align}
where the last line of \eqref{eq:case32} is due to the fact that the events are disjoint. 
Assuming that the probabilities of these two events are lower bounded by the same values, i.e., $(1-\vartheta)/2$, we have
the inequalities %$\text{Pr}\left[ (\Xi_N,t)\models\neg\varphi_1 \right]\ge (1-\vartheta)/2$ and
\begin{align}
&\text{Pr}\left[ (\Xi_N,t)\models\neg\varphi_1 \right]\ge (1-\vartheta)/2 ,  \label{case322}\\
&\text{Pr}\left[ (\Xi_N,t) \models\varphi_1\wedge\neg\varphi_2\right] \geq \frac{1-\vartheta}{2},
\label{case323}
\end{align}
which are in the form of inequalities discussed previously.
%\begin{align}
%&\hspace{-0.8cm}\text{Pr}\left[ (\Xi_N,t) \nvDash\varphi_1 \right] \geq \frac{1-\vartheta}{2} \Leftrightarrow \text{Pr}\left[ (\Xi_N,t) \models\varphi_1 \right] \leq \frac{1+\vartheta}{2}, \label{case321} \\
%&\hspace{-0.8cm}\text{Pr}\left[ ((\Xi_N,t) \models\varphi_1\wedge (\Xi_N,t) \nvDash\varphi_2) \right] \geq \frac{\delta_t}{2} \nonumber \\
% \Leftrightarrow & \text{Pr}\left[ ((\Xi_N,t) \nvDash\varphi_1\vee (\Xi_N,t) \models\varphi_2) \right] \leq 1-\frac{\delta_t}{2} \nonumber \\
% \Leftarrow &  \text{Pr}\left[ ((\Xi_N,t) \nvDash\varphi_1\right] + \text{Pr}\left[ (\Xi_N,t) \models\varphi_2) \right] \leq 1-\frac{\delta_t}{2} \nonumber\\
% \Leftrightarrow & \text{Pr}\left[ ((\Xi_N,t) \nvDash\varphi_1\right] \leq \frac{1}{2}(1-\frac{\delta_t}{2}) \nonumber \\
% &\text{and}\; \text{Pr}\left[ (\Xi_N,t) \models\varphi_2) \right] \leq \frac{1}{2}(1-\frac{\delta_t}{2})
% \label{case322}
%\end{align}
Note that Equations \eqref{eq:case31} to \eqref{case322} discuss the case of having 
conjunction of two STL formulas.
The results can be easily extended to conjunction of $n$ STL formulas by replacing $(1-\vartheta)/2$ with $(1-\vartheta)/n$.
%Similarly, in \eqref{case322}, devision by 2 should be replaced by the number of terms 
%that appear in the conjunction. \\
\medskip

\noindent
{\bf Case III: $\varphi = \varphi_1\mathcal{U}_{[a,b]}\varphi_2$}
The satisfaction $(\Xi_N,t)\models \varphi_1\mathcal{U}_{[a,b]}\varphi_2$ is equivalent to $\bigvee_{j=t+a}^{t+b}\psi_j$ with disjoint events
\begin{equation*}
\psi_j \!=\!\!\!\! \bigwedge_{i=t}^{t+a-1}(\Xi_N,i)\!\models\varphi_1 \!\!\!\bigwedge_{i=a+t}^{j-1}(\Xi_N,i)\!\models(\varphi_1\wedge\neg\varphi_2)\wedge(\Xi_N,j)\models\varphi_2.
\end{equation*}
Thus $\text{Pr}\left[  (\Xi_N,t)\models \varphi_1\mathcal{U}_{[a,b]}\varphi_2\right] \geq \vartheta$ 
is equivalent to $\sum_{j=t+a}^{t+b}\text{Pr}[\psi_j]\geq\vartheta$. Assuming the probabilities of events 
are lower bounded by the same values, we have
$\text{Pr}[\psi_j]\geq\vartheta/(b-a+1)$ for $j=a+t,\dots,b+t$, which again can be reduced as in Case II.

The second possible probabilistic constraint in Case III can be obtained as 
\begin{align}
\text{Pr}\left[  (\Xi_N,t)\models \varphi_1\mathcal{U}_{[a,b]}\varphi_2\right] \leq \vartheta %\nonumber \\
&\Leftrightarrow  \text{Pr}\left[ \bigvee_{j=a+t}^{b+t} \psi_j \right] \leq \vartheta \nonumber \\
&\Leftrightarrow \sum_{j=t+a}^{t+b}\text{Pr}[\psi_j] \leq \vartheta \nonumber \\
&\Leftrightarrow  \text{Pr}[\psi_j]\geq\vartheta/(b-a+1),
\end{align}
which can be again reduced as in Case II. Here also, we used the fact 
that $\psi_j$ consists of disjoint events %and hence, 
%$\text{Pr}( \mathcal A \vee \mathcal B) = \text{Pr}[\mathcal A] + \text{Pr}[\mathcal B]$ 
%for disjoint events $\mathcal A$ and $\mathcal B$, 
and we assume that he probabilities of events are lower bounded by the same 
value, i.e., by $\vartheta/(b-a+1)$, for $j=a+t,\dots,b+t$. 
%
%\Samira{In all these cases, the same procedure can be repeated recursively on the structure 
%of the sub-formulas until the obtained sub-formula is an atomic predicate.}
\end{proof}

%Note that in all the above cases, 
%we have not yet obtained a linear representation for 
%$\text{Pr}\left[ \rho^{\varphi}(\Xi_N,t) > 0 \right] \geq (\leq) 1-\delta_t$.
%%(and equivalently in \eqref{eq:prob5} and \eqref{eq:prob6}),  
%This can be done inductively as follows: in each expression of these equations, 
%if $\psi, \varphi_1$, and $\varphi_2$ are  
%defined as $\pi^{\alpha}$, then we are in Case I; if not, 
%%$\psi, \varphi_1$ and $\varphi_2$ are not linear predicates, 
%we repeat the same procedure as in Case II to IV, until the formula 
%we have is an atomic predicate $\pi^{\alpha}$ and then we 
%are in Case I again. Hence,
So far we have shown how to inductively reduce the chance constraint \eqref{eq:probconst} 
to inequalities of the form \eqref{eq:inductive_reduction} with atomic predicates.
In the rest of this section we discuss their corresponding linear inequalities for the two types 
of probability distributions considered in this paper.

\subsection{Arbitrary probability distributions with bounded support}
\label{sec:probtrans}
%Recall that, using the notion of robustness, 
%the chance constraint $\text{Pr}\left[\!{\Xi}(x(0),u_{\text{old}},w_{\text{realized}},\bar{u}(t),\bar{W}(t)) \!\models\! \varphi \! \right]$ 
%in \eqref{eq:finopt} is greater than or equal to %equivalent to  
%$\text{Pr}\left[\rho^{\varphi}({\Xi}(x(0),u_{\text{old}},w_{\text{realized}},\bar{u},\bar{W}),t) >0 \right]$.
%Hence, we replace the probabilistic constant by this lower bound such that 
%$\text{Pr}\left[\rho^{\varphi}({\Xi}(x(0),u_{\text{old}},w_{\text{realized}},\bar{u},\bar{W}),t) >0 \right] \geq 1-\delta_t$.
%Accordingly, in this section, we find a linear equivalence 
%for 
%\[
%\text{Pr}\left[\rho^{\varphi}({\Xi}(x(0),u_{\text{old}},w_{\text{realized}},\bar{u},\bar{W}),t) >0 \right]. 
%\]
%We refer to the stochastic run ${\Xi}(x(0),u_{\text{old}},w_{\text{realized}},\bar{u}(t),\bar{W}(t))$ 
%by $\bar{X}(t)$ in this section, for brevity. 

To transform the chance constraints into linear 
constraints in the case of having random variables with 
arbitrary probability distributions, we apply an approximation method based 
on the upper bound proposed by \cite{BouGou:16}. Let 
$Z_1,\ldots,Z_n$ be random variables with interval of 
bounded support $[a_i,b_i]$ and let $\mathbb{E}[Z_1],\ldots,\mathbb{E}[Z_n]$ denote their 
expected values belonging to the moment intervals $\mathbb{M}_{i}$ 
for $i=1,\ldots,n$. Define $Z = \sum_{i=1}^{n}Z_i$ and 
$\mathbb{E}(Z) = \sum_{i=1}^{n}\mathbb{E}[Z_i]$.  Using Chernoff-Hoeffding inequality, 
the following upper bound exists for any $\varsigma\geq 0$ \cite{Jan:04}
\begin{equation}
\text{Pr}\left[ Z - \mathbb{E}[Z] \leq -\varsigma \right] \leq \text{exp}\left( \frac{-\varsigma^2}{\nu\sum_{i=1}^{n}(b_i-a_i)^2} \right).
\label{eq:sir}
\end{equation}
where $\nu>0$ is a constant. If $Z_1,\ldots,Z_n$ are dependent, then the 
inequality applies with a constant $\nu = \chi(\hat{G})/2$, where $\hat{G}$ 
denotes the indirected dependency graph of $Z_1,\ldots,Z_n$ and $\chi(\hat{G})$ 
is the chromatic number of the graph $\hat{G}$ defined as the minimum number 
of colors required to color $\hat{G}$. For the independent case, $\chi(\hat{G})=1$. 
The expression for the right tail probability is derived identically.
For more details, the reader is referred to \cite{BouGou:16}.

%Here, we discuss Case I for the case of having random variables with 
%known intervals of support and first moment intervals.
%\\
%{\bf Case I}: $\varphi := \pi^{\alpha}$ \\
%In this case, since by assumption, $\alpha$ is a
Consider the chance constraints \eqref{eq:inductive_reduction} with $\varphi = \{\alpha\ge 0\}$.
Since $\alpha$ is an
affine function of random state variables, it is a 
random variable itself with the following interval of support and 
moment interval
\begin{equation}
\begin{array}{c}
I_{\alpha(X(t))} = [\tilde{a}_{t} + \tilde{C}_{t},\tilde{b}_{t}+\tilde{C}_{t}] \\
\mathbb{M}_{\alpha(X(t))} = [\tilde{c}_{t}+\tilde{C}_{t},\tilde{d}_{t}+\tilde{C}_{t}]
\end{array}
\label{eq:interval3}
\end{equation} 
where for $t=0,\ldots,N$, we have $\tilde{a}_{t},\tilde{b}_{t},\tilde{c}_{t}$ and $\tilde{d}_{t}$ are weighted sum of 
$\bar{a}_t,\bar{b}_t,\bar{c}_t,\bar{d}_t$ related to the interval of support and moment 
interval of random variables $X(t)$ (cf. \eqref{eq:interval}), and $\tilde{C}_{t}$ is a linear expression of 
input variables. 

%\textcolor{red}{Fix the issue related to the fact that moments of the random 
%variables may not be greater than 0!}. 
Let $\varsigma = \mathbb{E}\left[ \alpha(X(t))\right]$; we can directly use \eqref{eq:sir} as
\begin{align}
\hspace{-1cm}\text{Pr}\left[ (\Xi_N,t)\models\pi^{\alpha} \right] \geq &1-\delta_t \Leftrightarrow \text{Pr}\left[ \alpha(X(t)) > 0 \right] \geq 1-\delta_t \nonumber \\
\Leftrightarrow & \text{Pr}\left[ \alpha(X(t)) \leq 0 \right] \leq \delta_t  \nonumber \\
\Leftarrow & \text{exp}\left( \frac{-\varsigma^2}{\nu\sum_{t=1}^{N}(\tilde{b}_t-\tilde{a}_t)^2} \right) \leq \delta_t  \nonumber \\
\Leftrightarrow & \frac{-\varsigma^2}{\nu\sum_{t=1}^{N}(\tilde{b}_t-\tilde{a}_t)^2} \leq \log(\delta_t) \nonumber \\
\Leftrightarrow & -\varsigma^2 \leq \nu\log(\delta_t) \sum_{t=1}^{N}(\tilde{b}_t-\tilde{a}_t)^2 \nonumber \\
\Leftarrow & \varsigma \geq \sqrt{-\nu\log(\delta_t)  \sum_{t=1}^{N}(\tilde{b}_t-\tilde{a}_t)^2} %\nonumber \\
%\Leftrightarrow & \tau \geq \sqrt{-\nu\log(\delta_t)  \sum_{t=1}^{N}(\bar{b}_t-\bar{a}_t)^2} \; \nonumber \\
%\text{or}&\;\;  \tau \leq -\sqrt{-\nu\log(\delta_t)  \sum_{t=1}^{N}(\bar{b}_t-\bar{a}_t)^2}
\label{eq:case1}
\end{align} 
%where $\tau = -\mathbb{E}\left[ \rho^{\varphi}(\Xi_N,t)\right]$ and we can replace 
%each of the inequalities in the last two lines of \eqref{eq:case1} by 
Note that since $\delta_t\in(0,1)$, we 
have $\log(\delta_t)<0$; hence, by multiplying both sides of the inequality by -1 in 
line 5 of \eqref{eq:case1}, the expression $-\log(\delta_t) \cdot \sum_{t=1}^{N}(\tilde{b}_t-\tilde{a}_t)^2$ 
becomes a positive number, and hence, its square root is a real number. 
Note also that the last inequality is due to the fact that $\varsigma\geq 0$. 
Hence, we can replace $\varsigma$ in the last inequality of \eqref{eq:case1} by 
the lower bound % upper bound and lower bound 
of its moment interval in \eqref{eq:interval3}, 
i.e., with $\tilde{c}_{t}+\tilde{C}_{t}$, % $\bar{d}_{t}+\bar{C}_{t}$ and $\bar{c}_{t}+\bar{C}_{t}$
which is a linear expression in the input variables. 

Consequently, in this case, 
the chance constraint in \eqref{eq:finopt} can be replaced 
by %one of the following linear constraints 
\begin{equation}
%\bar{d}_{t}+\bar{C}_{t} \leq -\sqrt{-\nu\log(\delta_t) \cdot \sum_{t=1}^{N}(\bar{b}_t-\bar{a}_t)^2} \label{eq:sir1}\\
\tilde{c}_{t}+\tilde{C}_{t} \geq \sqrt{-\nu\log(\delta_t) \cdot \sum_{t=1}^{N}(\tilde{b}_t-\tilde{a}_t)^2}. 
\label{eq:sir2}
\end{equation}
For the second type of probabilistic inequality (cf. \eqref{eq:inductive_reduction}), we can 
again use \eqref{eq:sir} for the right tail probability; hence we have
\begin{align}
\text{Pr}\left[ (\Xi_N,t)\models\pi^{\alpha} \right] &\leq 1-\delta_t \nonumber \\
\Leftarrow& \text{Pr}\left[ \alpha(X(t)) \geq 0 \right] \leq 1-\delta_t \nonumber \\
\Leftarrow & \text{exp}\left( \frac{-\varsigma^2}{\nu\sum_{t=1}^{N}(\tilde{b}_t-\tilde{a}_t)^2} \right) \leq 1-\delta_t, 
\label{case11}
\end{align}
and then following the same steps as in \eqref{eq:case1}, we obtain the same linear expression 
for the chance constant as in \eqref{eq:sir2} by only replacing $\delta_t$ by $1-\delta_t$ in the 
related expressions.
\subsection{Normal distribution}
\label{sec:prob}
To transform the chance constraints into linear constraints in the case of having 
normally distributed random variables, we use the quantile of the 
normal distribution. By definition, for a normally distributed random variable $x$ with mean 
$\mu$ and standard deviation $\sigma$,
\begin{align}
&\text{Pr}[x\leq b] \leq \delta_t \Leftrightarrow F^{-1}(\delta_t)\geq b \Leftrightarrow \mu + \sigma \phi^{-1}(\delta_t) \geq b
\label{quantile1} \\
&\text{Pr}[x\leq b] \geq \delta_t \Leftrightarrow F^{-1}(\delta_t)\leq b \Leftrightarrow \mu + \sigma \phi^{-1}(\delta_t) \leq b
\label{quantile2}
\end{align}
where $F^{-1}$ denotes the inverse of the cumulative distribution function 
or the quantile function and $\phi^{-1}$ is the inverse of the error function 
of a normally distributed random variable. 

%Here also, we consider Case I separately. \\
%{\bf Case I}: $\varphi := \pi^{\alpha}$ \\
Recall the chance constraints \eqref{eq:inductive_reduction} with $\varphi = \{\alpha\ge 0\}$.
Since $\alpha$ is an 
affine function of normally distributed state variables, it is also 
normally distributed with appropriately defined mean $\mu_t$ and 
variance $\sigma_t^2$. %where $\mu_t = \sum_{i=1}^n\xi_i + \mu_i$ 
%and $\sigma_t^2 = \sum_{i=1}^n\sigma_i^2$ where $\mu_i$ and $\sigma_i,\;i=1,\dots,n$ 
%are mean and variance of the state variables, which are present in $\alpha$. 
Hence, we can directly use \eqref{quantile1} and \eqref{quantile2} as
\begin{align}
\text{Pr}\left[ (\Xi_N,t)\models\pi^{\alpha} \right] \geq 1-\delta_t \Leftrightarrow& \text{Pr}\left[ \alpha(X(t)) > 0 \right] \geq 1-\delta_t \nonumber \\
\Leftrightarrow & \text{Pr}\left[ \alpha(X(t)) \leq 0 \right] \leq \delta_t \nonumber\\
\Leftrightarrow& F^{-1}(\delta_t)\geq 0 \nonumber \\
\Leftrightarrow & \mu_{t} + \sigma_t \phi^{-1} (\delta_t) \geq 0,
\label{eq:prob0} \\
\text{Pr}\left[ (\Xi_N,t)\models\pi^{\alpha} \right] \leq 1-\delta_t \Leftrightarrow&\text{Pr}\left[ \alpha(X(t)) > 0 \right] \leq 1-\delta_t \nonumber \\
\Leftrightarrow & \text{Pr}\left[  \alpha(X(t))  \leq 0 \right] \geq \delta_t \nonumber\\
\Leftrightarrow & F^{-1}(\delta_t)\leq 0 \nonumber \\
\Leftrightarrow & \mu_{t} + \sigma_t \phi^{-1} (\delta_t) \leq 0.
\label{eq:prob1}
\end{align} 
Therefore, the chance constraint can be replaced 
by the equivalent linear constraint \eqref{eq:prob0} or \eqref{eq:prob1}, 
depending on the type of constraint we have.

\section{Experimental Results}
\label{sec:exam}

We now apply our controller synthesis approach to the 
room temperature control in a building. The details 
of the thermal model can be found in \cite{Maa:13,Raman15}, and 
is briefly explained here for clarity.
The temperature of room 
$r$ is denoted by $T_{r}$ and the wall and the temperature of the wall 
between the room and its surrounding $j$ (e.g. other rooms) are denoted by $w_{j}$ and 
$T_{w_{j}}$, respectively. Dynamics of the temperature of wall $w_{j}$ and 
room $r$ can be written as \cite{Maa:13}
\begin{align}
C_{j}^{w}\frac{dT_{w_{j}}}{dt} &= \sum_{k\in\mathcal{N}_{w_{j}}}\frac{T_{r,k}-T_{w_{j}}}{R_{j,k}}  + r_j\alpha_{j}A_{w_{j}}Q_{\text{rad}j} \label{eq:examp1}\\
C_{j}^{r}\frac{dT_{r}}{dt} &= \sum_{k\in\mathcal{N}_{r}}\frac{T_{k}-T_{r}}{R_{k}}  + \dot{m}_{r}c_a(T_{s} - T_{r}) \nonumber\\
&\;\;\;\;\;+ w_i\tau_{w}A_{\text{win}}Q_{\text{rad}}+\dot{Q}_{\text{int}}
\label{eq:examp2}
\end{align}
where $C_{j}^{w},\alpha_{j}$ and $A_{w_{j}}$ are heat capacity, a radiative heat
absorption coefficient, and the area of $w_{j}$, respectively. The total thermal 
resistance between the centerline of wall $j$ and the side of the wall on which 
node $k$ is located is denoted by $R_{j_k}$. 
The radiative heat flux density on $w_{j}$ is denoted by $Q_{\text{rad}j}$, the
set of all neighboring nodes to $w_{j}$ is denoted by $\mathcal{N}_{w_{j}}$, and $r_{j}$ is a wall identifier,
which equals 0 for internal walls and 1 for peripheral walls,
where $j$ is the outside node. The temperature, heat capacity and air mass flow 
into room $r$ are denoted by $T_{r} , C_{j}^{r}$ and $\dot{m}_{r}$, respectively; 
$c_a$ is the specific heat capacity of air, and $T_{s}$ is
the temperature of the supply air to room $r$, $w$ is a window
identifier, which equals 0 if none of the walls surrounding
room $r$ have windows, and 1 if at least one of them does, $\tau_{w}$
is the transmissivity of the glass of the window in room $r$, $A_{\text{win}}$ is the
total area of the windows on walls surrounding room $r$, $Q_{\text{rad}}$
is the radiative heat flux density per unit area radiated to room
$r$, and $\dot{Q}_{\text{int}r}$ is the internal heat generation in room $r$. Finally, 
$\mathcal{N}_{r}$ is the set of neighboring room nodes for room $r$. Further details
on this thermal model can be found in \cite{Maa:13}.

As such, the heat transfer equations for each wall and room $r$ is in the form of nonlinear differential equation.
After linearization and time-discretization, the model of the system becomes in the form of dynamical equation
% \eqref{eq:discsys}
%\begin{align*}
%\dot{X}(t) = f(X(t),u(t),W(t)),
%\end{align*}
%with the following discrete-time linearized model 
\begin{equation*}
X(t+1) = A X(t) + B u(t) + W(t),
%\label{eq:example}
\end{equation*}
where $X\in\mathbb{R}^{n}$ is the state vector representing the 
temperature of the walls and the rooms and $u\in\mathbb{R}^{m}$ is 
the input vector representing the air mass flow rate and discharge air 
temperature of conditioned air into each thermal zone.
Matrices $A$ and $B$ are obtained by time discretization of dynamics 
of the thermal model 
\eqref{eq:examp1}-\eqref{eq:examp2} 
with a sampling time 
of $t_s = 30$ minutes.
The disturbance $W(\cdot)$ aggregates various \emph{unmodeled} 
dynamics and the uncertainty in physical variables such as the outside temperature, internal heat generation 
and radiative heat flux density.
The statistics of $W(\cdot)$ can be estimated using historical
data \cite{Maa:13}.

In this example, we only control the temperature of one room and 
include the temperature of the neighboring rooms as part of the 
disturbance signal $W(t)$.
We also assume that there is a reference for the disturbance 
signal, denoted by $w_r(t)$, and the reference is perturbed by 
independent and identically distributed random vectors $e(t)\sim\mathcal{N}(0,\mathbb I_n)$, %\;i\in\mathbb N_n$, 
i.e., the disturbance is
$W(t) = w_r(t) + e(t)$, 
which is normally distributed with mean $\mu_{t}=w_r(t)$ and 
identity covariance matrix $\Sigma = \mathbb I_n$.
%\VP{VP: $W_r(t)$ should be $w_r(t)$, or some smallcase letter since the reference signal is
%not stochastic.}
%\Samira{Considering the definition of $W(t)$ in this example, 
%it can take both positive and negative quantities and hence, normal distribution would be 
%an appropriate choice for the distribution of the disturbance.}

%We assume that the disturbance signal 
%is normally distributed. 
In contrast to \cite{Raman15}, which considers deterministic disturbances from
a bounded set, we consider stochastic disturbances and we maximize the robustness 
of satisfaction of the STL specifications in the presence of such disturbance.
Accordingly, we handle chance constraints and include the expected 
value of the robustness function in the objective function.

We consider a signal $\text{occ}:\mathbb N\rightarrow\{-1,1\}$ representing 
the room occupancy; $\text{occ}(t) = 1$ if the room is occupied at time $t$ 
and $\text{occ}(t) = -1$ otherwise. 
This signal is assumed to be known for the entire simulation period.  
The MPC prediction horizon $N$ is chosen to be 
$24$, representing 12 hours monitoring of the room temperature. We select 
$\delta = 0.1$ so that the obtained control input 
provides confidence level of $90\%$ on the
satisfaction of the desired behavior.
We are interested in keeping the room temperature above 
a reference temperature $T_r$ when the room is occupied; thus the specification is
\begin{equation*}
\psi = \G_{[0,N]}\big(\text{occ}(t) = 1\,\,\rightarrow
\,\, X(t) > T_r\big).
\end{equation*}
At each time instant $0\leq t < N$, the optimization 
problem \eqref{eq:finopt} obtains an optimal control input 
$\tilde u_{opt}(t:N) = [u_{opt}(t),\ldots,u_{opt}(N-1)]$ that minimizes
\begin{equation*}
\mathbb{E}[-\rho^{\psi}(\bar X(0:t:N),0)] + \sum_{k=t}^{N-1}||u(k)||_1,
\end{equation*}
where the robustness function is defined as
\begin{equation*}
\rho^{\psi}(\bar X(0:t:N),0) = \min\{ X(\tau) - T_r\  \mid \tau\in[0,N], \text{occ}(\tau)>0 \}.
\end{equation*}
%$\bar X(0:t:N) = [x_0,\ldots,x_t,X(t+1),\ldots,X(N)]$ is a finite run of 
%state variables such that $x_0,x_1,\ldots,x_t$ are the observed states and $X(t+1),\ldots,X(N)$ 
%are the random state variables at time $t$,
The chance constraint \eqref{eq:probconst} is defined with the same specification $\varphi = \psi$.
%as the probability that 
%a finite stochastic run of system \eqref{eq:example} satisfies $\psi$ is greater 
%than or equal to $1-\delta$.
We approximate $\mathbb{E}[-\rho^{\psi}(\tilde X(0:t:N),0)]$ 
using the upper bound \eqref{eq:app2} and
%we approximate the chance constraints by its lower bound $\text{Pr}[\rho^{\psi}(\bar X(0:t:N),t)>0] \geq 1-\delta$ and then
transform the chance constraint \eqref{eq:probconst}
into linear inequalities using the approach of Section \ref{sec:const}. 
We also assume that inputs are bounded, i.e., for each $0\leq t<N$, we have $0 \leq u(t) \leq 380$.

The simulations are done using {\sc Matlab} R2014b on a 2.6 GHz Intel Core i5 processor and 
the optimizations are solved using {\it fmincon} solver in {\sc Matlab}. 
we perform $n_s = 200$ simulations in order to check the satisfiability of the STL specifications with a probability 
greater than or equal to $0.9$.   
Figure~\ref{fig:fig1} shows the results of these $200$ simulations.
The top plot shows the occupancy signal and the middle plot 
illustrates the average, minimum, and maximum of the 
obtained room temperatures over 12 hours
as well as the room reference temperature $T_r$ in Fahrenheit. 
The controller ensures that the room temperature goes above the reference temperature 
$T_r$ once the occupancy signal is 1 and stays there as 
long as the room is occupied. The minimum and maximum bounds on the room 
temperature shows that the specifications have never  
been violated in these $200$ simulations. The bottom plot shows the average, minimum, and maximum of the 
air flow rate in
%$\text{ft}^3/\text{min}$,
$\left[\frac{\text{ft}^3}{\text{min}}\right]$, 
which indicates that the input constraint is not violated.

\begin{figure}[t]
      \centering
      \includegraphics[width=\linewidth]{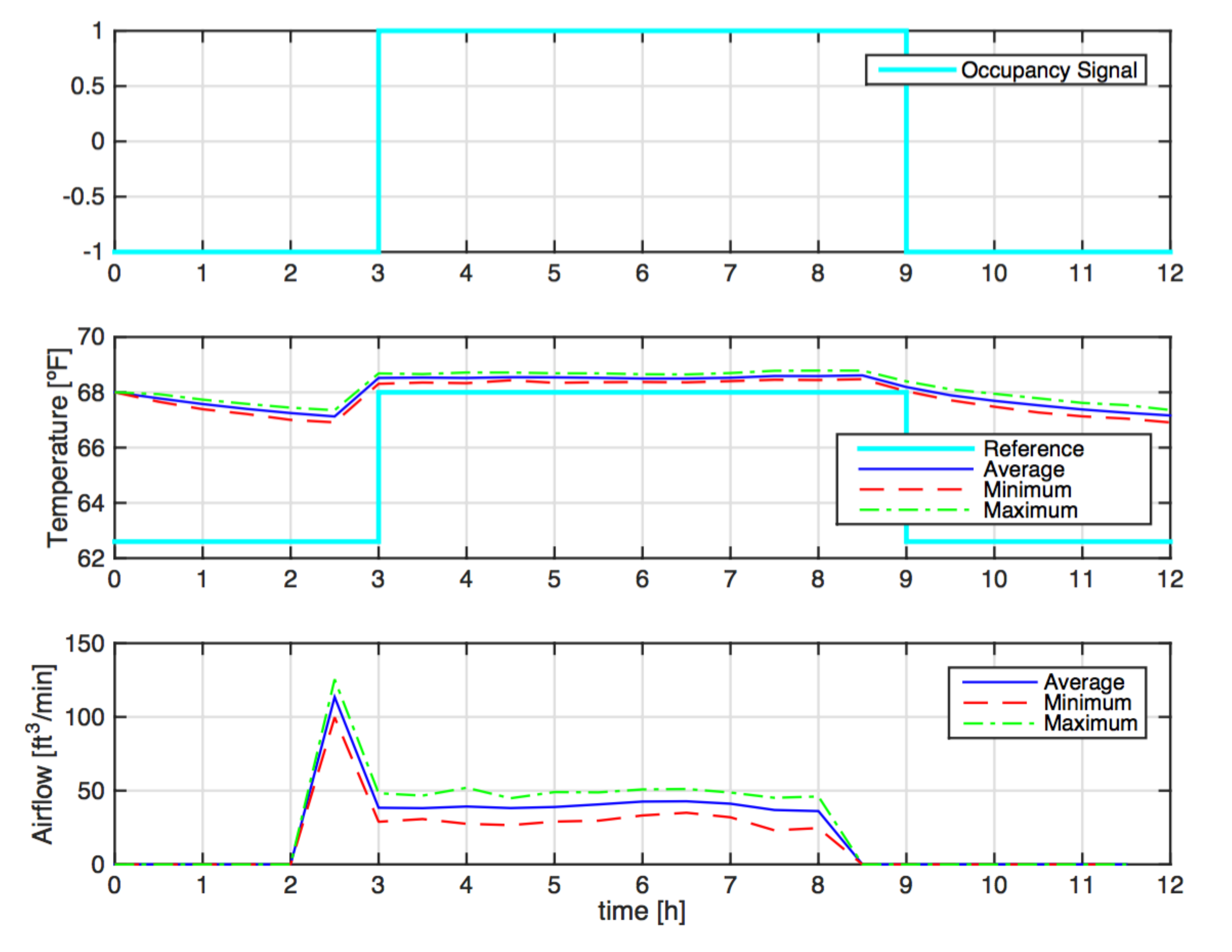} 
      \caption{Controlling the room temperature using SHMPC in the presence of normally 
      distributed disturbance and STL constraints.}
\label{fig:fig1}
\end{figure}

%The reason for picking 200 simulations is related to the confidence level we 
%would like to have to show the feasibility of the optimization problem 
%\eqref{eq:finopt} after applying the approximations presented 
%in Sections \ref{sec:obj} and \ref{sec:const}.
Note that all these $n_s = 200$ runs result in feasible solutions, which gives a confidence 
bound on the feasibility of the original problem as follows. 
%According to \cite{CloPea:34}, 
Since all the $n_s$ runs of the simulation are feasible, we can claim that the original problem is
also feasible with probability at least $(\beta/2)^{1/n_s}$, $\beta\!\in\!(0,1)$, with confidence level $1\!-\!\beta$ \cite{CloPea:34};
hence, having $200$ runs being all feasible, % according to our simulation results, 
the optimization problem \eqref{eq:finopt} is also feasible with probability $0.98$ with confidence level $0.95$. 

To further illustrate the performance of the proposed method, we compare our SHMPC approach with the 
robust MPC (RMPC) approach of \cite{Raman15}, 
in which the disturbance belongs to a bounded polyhedral set.
%% and with the one of the open-loop optimal control. 
Note that RMPC approach is not directly applicable to unbounded uncertainties. Therefore, 
in the optimization procedure, we truncate a normally distributed disturbance 
in the $2\sigma$ interval such that $e(t)\in[-1,1]$. 
Further, we solve the RMPC optimization problem using Monte Carlo sampling.

\begin{table}[t]
\caption{Comparison of the statistics of the fan energy consumption using different control approaches.}
 \begin{tabular}{|c|c|c|}
  \hline
   \text{Computational}  & \text{Fan energy } & Average \\
   \text  Methods& consumption [kWh] & computation time [s]  \\ [0.5ex]
   \hline
   \text{Open-loop OC} & $1337.016$ &  3.9277   \\ [0.5ex]
   \text{RMPC}&$\mu_1 =12.2216 $, $\sigma_1 = 0.045 \mu_1$ &   33.4891 \\ [0.5ex]
   \text{SHMPC} & $\mu_2 = 2.5101$, $\sigma_2 = 0.104\mu_2$ &  19.3622  \\ [0.5ex]
     \hline
 \end{tabular}
 \label{tab:tab1}
\end{table}

The total fan energy consumption is proportional to the cubic of 
the airflow.
Table~\ref{tab:tab1} shows the total fan energy consumption and the computation 
time for the three approaches.
For RMPC and SHMPC, we report the average and standard deviation of total energy 
consumption using the sum of cubes of the optimal input sequences corresponding to 
the $200$ simulations. Also, for these two approaches, we report the average computation 
time over the $200$ simulations. 
Comparing statistics of these two approaches is essential because of the chance 
constraints in SHMPC and the Monte Carlo sampling based optimization in RMPC.
The energy consumption using open-loop optimal control (OC) is very high, 
comparing to both RMPC and SHMPC. This is due to the fact that the open-loop 
strategy computes the solution of optimization problem \eqref{eq:finopt} only once 
and hence, the computation time is smaller compared to the two other methods.
As a result, the input sequence has an aggressive behavior to make 
sure that it reacts in time to the changes happening in the system.
Since RMPC is more conservative compared to SHMPC, the average energy consumption is much higher for
the RMPC controller compared to the SHMPC controller:
the SHMPC controller achieves a $80\%$ reduction of total energy consumption on average compared to RMPC.

\section{Conclusions}
\label{sec:con}
In this paper, we presented shrinking horizon model predictive control (SHMPC) for 
stochastic discrete-time linear systems with signal temporal logic (STL) specifications. 
Our aim was to obtain an optimal control sequence that guarantees the satisfaction of 
STL specifications with a probability greater than a certain level. By assumption, 
the stochastic disturbance signal had an arbitrary probability distribution with a bounded support and the only 
available information related to this distribution is the intervals of support and the moment 
intervals of each component of the disturbance signal. 
Using an existing approximation technique, we showed that the chance constraints 
could be approximated by an upper bound, which resulted in having approximate linear 
constraints for the chance constraints. Moreover, in the case of having the state costs 
and/or the robustness function related to the degree of satisfaction of the specifications 
by the state trajectory, their expected value can be also approximated using the moment intervals 
of components of the disturbance signal. As an additional case, we further considered disturbances that 
are normally distributed and we showed that the chance constraints in this case can be 
replaced by the quantile expressions which are linear in the input variables. 
%We also showed that in this case, the expected value of the robustness function can be 
%approximated by its upper bound and hence, the resulting optimization problem is convex.
In the end, in an example, we applied the proposed method to control a HVAC system.


%%%%%%%%%%%%%%%%%%%%%%%%%%%%%%%%%%%%%%%%%%%%%%%%%%%%%%%%%%%%%%%%%%%%%%%%%%%%%%%%

%%%%%%%%%%%%%%%%%%%%%%%%%%%%%%%%%%%%%%%%%%%%%%%%%%%%%%%%%%%%%%%%%%%%%%%%%%%%%%%%

%%%%%%%%%%%%%%%%%%%%%%%%%%%%%%%%%%%%%%%%%%%%%%%%%%%%%%%%%%%%%%%%%%%%%%%%%%%%%%%%


%\section*{ACKNOWLEDGMENTS}
%Do we have any?!!!

\bibliographystyle{plain}        % Include this if you use bibtex 

%\providecommand{\noopsort}[1]{}\providecommand{\sahinidummy}{}\providecommand{\Zs}{Zs}

%\bibliography{strings,references}  

\begin{thebibliography}{10}
	
	\bibitem{AbrSte:64}
	M.~A. Abramowitz and I.~Stegun.
	\newblock {\em Handbook of Mathematical Functions}.
	\newblock National Bureau of Standards, US Government Printing Office,
	Washington DC, 1964.
	
	\bibitem{AtaSav:05}
	A.~Atamt{\"{u}}rk and M.~W.~P. Savelsbergh.
	\newblock Integer-programming software systems.
	\newblock {\em Annals of Operations Research}, 140(1):67--124, November 2005.
	
	\bibitem{BemHee:02-003}
	A.~Bemporad, W.~P.~M.~H. Heemels, and B.~{De Schutter}.
	\newblock On hybrid systems and closed-loop {MPC} systems.
	\newblock {\em IEEE Transactions on Automatic Control}, 47(5):863--869, May
	2002.
	
	\bibitem{BouGou:16}
	O.~Bouissou, E.~Goubault, S.~Putot, A.~Chakarov, and S.~Sankaranarayanan.
	\newblock Uncertainty propagation using probabilistic affine forms and
	concentration of measure inequalities.
	\newblock In {\em Proceedings of Tools and Algorithms for Construction and
		Analysis of Systems (TACAS)}, volume TBA of {\em Lecture Notes in Computer
		Science}, page TBA. Springer-Verlag, 2016.
	
	\bibitem{BueEng:00}
	B.~B\"{u}eler, A.~Enge, and K.~Fukuda.
	\newblock Exact volume computation for convex polytopes: A practical study.
	\newblock In G.~Kalai and G.M. Ziegler, editors, {\em Polytopes --
		Combinatorics and Computation}, pages 131--154. Birk\"auser Verlag, Basel,
	Switzerland, 2000.
	
	\bibitem{CloPea:34}
	C.~Clopper and E.~S. Pearson.
	\newblock The use of confidence or fiducial limits illustrated in the case of
	the binomial.
	\newblock {\em Biometrika}, 26(4):404--413, 1934.
	
	\bibitem{DavRab:84}
	P.~J. Davis and P.~Rabinowitz.
	\newblock {\em Methods of Numerical Integration}.
	\newblock Academic Press, New York, 2nd edition, 1984.
	
	\bibitem{DeSvan:99-10}
	B.~{De Schutter} and T.~{van den Boom}.
	\newblock Model predictive control for max-plus-linear discrete event systems.
	\newblock {\em Automatica}, 37(7):1049--1056, July 2001.
	
	\bibitem{DeSvan:02-004}
	B.~{De Schutter} and T.~J.~J. {van den Boom}.
	\newblock {MPC} for continuous piecewise-affine systems.
	\newblock {\em Systems \& Control Letters}, 52(3--4):179--192, July 2004.
	
	\bibitem{fainekos2009robustness}
	G.~E. Fainekos and G.~J. Pappas.
	\newblock Robustness of temporal logic specifications for continuous-time
	signals.
	\newblock {\em Theoretical Computer Science}, 410(42):4262--4291, 2009.
	
	\bibitem{FarMaj:16}
	S.~S. Farahani, R.~Majumdar, V.S. Prabhu, and S.~Esmaeil~Zadeh Soudjani.
	\newblock Shrinking horizon model predictive control with chance-constrained
	signal temporal logic specifications.
	\newblock In {\em To appear in the proceedings of American Control Conference
		(ACC) 2017}.
	
	\bibitem{FarRam:15}
	S.~S. Farahani, V.~Raman, and R.~M. Murray.
	\newblock Robust model predictive control for signal temporal logic synthesis.
	\newblock In {\em Proceedings of the IFAC Conference on Analysis and Design of
		Hybrid Systems}, Atlanta, Georgia, October 2015.
	
	\bibitem{Farvan:16}
	S.~S. Farahani, T.~{van den Boom}, H.~van~der Weide, and B.~{D}e Schutter.
	\newblock An approximation method for computing the expected value of
	max-affine expressions.
	\newblock {\em European Journal of Control}, 27:17--27, January 2016.
	
	\bibitem{GroOca:14}
	J.~M. Grosso, C.~Ocampo-Martinez, V.~Puig, and B.~Joseph.
	\newblock Chance-constrained model predictive control for drinking
	waternetworks.
	\newblock {\em Journal of process control}, 24:504--516, 2014.
	
	\bibitem{Maa:13}
	M.~Maasoumy Haghighi.
	\newblock {\em Controlling Energy Efficient Buildings in the Context of Smart
		Grid: A Cyber Physical System Approach}.
	\newblock PhD thesis, University of California, Berkeley, Berkeley, California,
	December 2013.
	
	\bibitem{Jan:04}
	S.~Janson.
	\newblock Large deviations for sums of partly dependent random variables.
	\newblock {\em Random Structures Algorithms}, 24(3):234--248, 2004.
	
	\bibitem{JhaRam:16}
	S.~Jha and V.~Raman.
	\newblock Automated synthesis of safe autonomous vehicle control under
	perception uncertainty.
	\newblock In {\em 8th International Symposium of NASA Formal Methods}, Lecture
	Notes in Computer Science, pages 117--132, Minneapolis, Minnesota, June 2016.
	
	\bibitem{req_mining_hscc2013}
	X.~Jin, A.~Donz{\'{e}}, J.~V. Deshmukh, and S.~A. Seshia.
	\newblock Mining requirements from closed-loop control models.
	\newblock In {\em Hybrid Systems: Computation and Control, {HSCC} 2013, April
		8-11, 2013, Philadelphia, PA, {USA}}, pages 43--52, 2013.
	
	\bibitem{Las:98}
	J.~B. Lasserre.
	\newblock Integration on a convex polytope.
	\newblock {\em Proceedings of the American Mathematical Society},
	126(8):2433--2441, August 1998.
	
	\bibitem{LazHee:06}
	M.~Lazar, M.~Heemels, S.~Weiland, and A.~Bemporad.
	\newblock Stability of hybrid model predictive control.
	\newblock {\em IEEE Transactions on Automatic Control}, 15(11):1813--1818,
	2006.
	
	\bibitem{JLiTRa:05}
	J.~T. Linderoth and T.~K. Ralphs.
	\newblock Noncommercial software for mixed-integer linear programming.
	\newblock In J.~Karlof, editor, {\em Reinforcement Learning: State-Of-The-Art},
	CRC Press Operations Research Series, pages 253--303. 2005.
	
	\bibitem{Mac:02}
	J.~M. Maciejowski.
	\newblock {\em Predictive Control with Constraints}.
	\newblock Prentice Hall, Harlow, England, 2002.
	
	\bibitem{MalNic:04}
	O.~Maler and D.~Nickovic.
	\newblock Monitoring temporal properties of continuous signals.
	\newblock In {\em In FORMATS/FTRTFT}, pages 152--166, 2004.
	
	\bibitem{Raman15}
	V.~Raman, A.~Donz{\'{e}}, D.~Sadigh, R.~M. Murray, and S.~A. Seshia.
	\newblock Reactive synthesis from signal temporal logic specifications.
	\newblock In {\em Hybrid Systems: Computation and Control, {HSCC} 2015,
		Seattle, WA, USA, April 14-16, 2015}, pages 239--248, 2015.
	
	\bibitem{SadKap:16}
	D.~Sadigh and A.~Kapoor.
	\newblock Safe control under uncertainty with probabilistic signal temporal
	logic.
	\newblock In {\em Proceedings of Robotics: Science and Systems Conference},
	AnnArbor, Michigan, June 2016.
	
	\bibitem{SchNik:99}
	A.~T. Schwarm and M.~Nikolaou.
	\newblock Chance-constrained model predictive control.
	\newblock {\em AIChE Journal}, 45(8):1743--1752, 1999.
	
\end{thebibliography}

\begin{IEEEbiography}[{\includegraphics[width=1in,height=1.25in,clip,keepaspectratio]{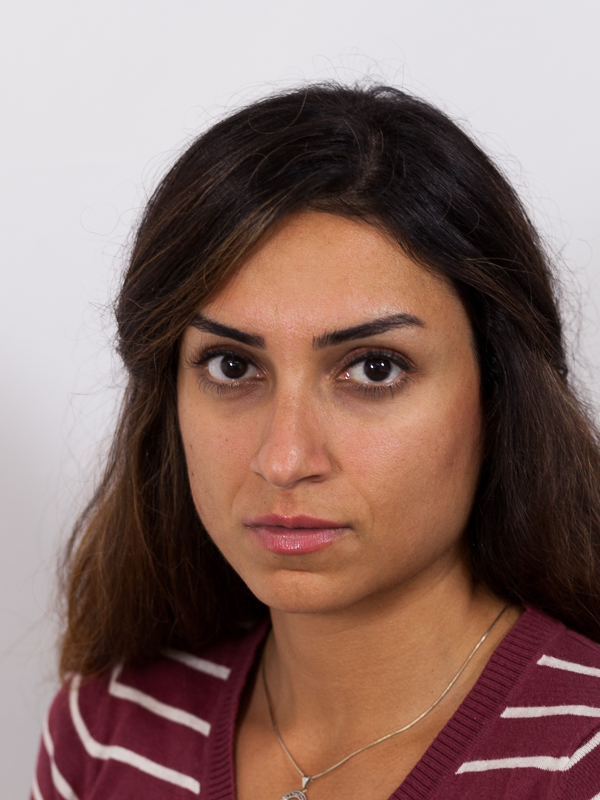}}]{Samira S. Farahani}
is a postdoctoral researcher at Max Planck Institute for Software Systems. Her research interests are 
in control synthesis for stochastic hybrid systems, temporal logic and formal methods, and discrete-event systems. Prior to this position, 
she was a postdoctoral researcher at CMS department, California Institute of Technology, and at Energy and 
Industry Department, Delft University of Technology. Dr. Farahani has 
received her PhD degree in Systems and Control and her MSc. degree in Applied Mathematics, both from Delft 
University of Technology, the Netherlands, in 2012 and 2008, respectively. She obtained her BSc. degree in Applied Mathematics from 
Sharif University of Technology, Iran, in 2005.
\end{IEEEbiography}
 
\begin{IEEEbiography}[{\includegraphics[width=1in,height=1.25in,clip,keepaspectratio]{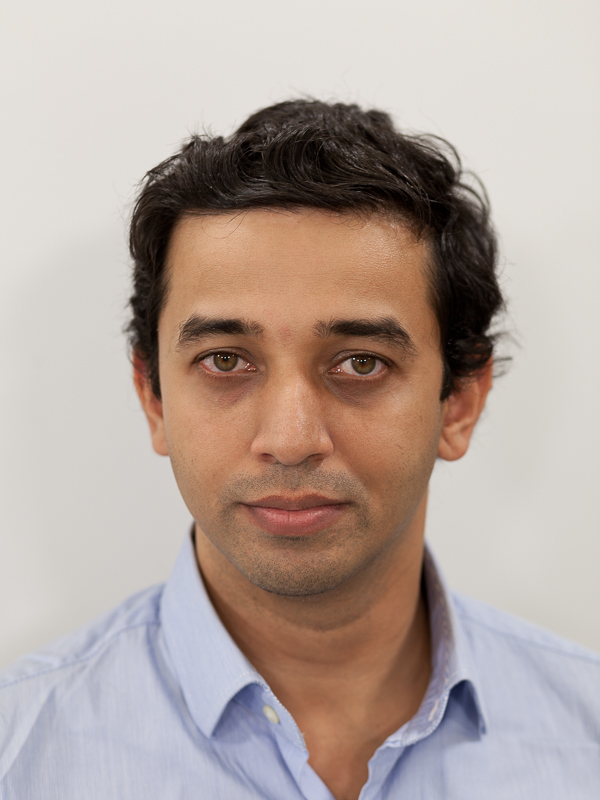}}]{Rupak Majumdar}
is a Scientific Director at the Max Planck Institute for Software Systems. His research interests are in the verification and control of reactive, real-time, hybrid, and probabilistic systems, software verification and programming languages, logic, and automata theory. Dr. Majumdar received the President's Gold Medal from IIT, Kanpur, the Leon O. Chua award from UC Berkeley, an NSF CAREER award, a Sloan Foundation Fellowship, an ERC Synergy award, "Most Influential Paper" awards from PLDI and POPL, and several best paper awards (including from SIGBED, EAPLS, and SIGDA). He received the B.Tech. degree in Computer Science from the Indian Institute of Technology at Kanpur and the Ph.D. degree in Computer Science from the University of California at Berkeley. 
\end{IEEEbiography}

\begin{IEEEbiography}[{\includegraphics[width=1in,height=1.25in,clip,keepaspectratio]{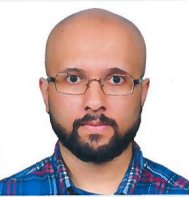}}]{Vinayak Prabhu}
did his PhD in the department of Electrical Engineering and Computer Sciences
at the University of California, Berkeley (USA) in 2008, and is currently
a postdoctoral researcher at Max Planck Institute for Software Systems.
His research interests are in modelling, analysis, verification, and
control, of reactive, real-time, hybrid, and multi-agent systems.
\end{IEEEbiography}

\begin{IEEEbiography}[{\includegraphics[width=1in,height=1.25in,clip,keepaspectratio]{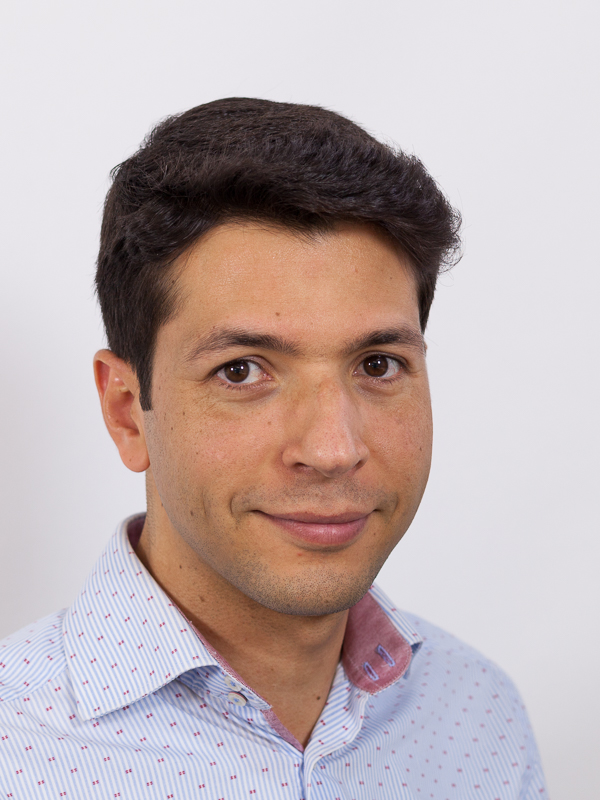}}]{Sadegh Esmaeil Zadeh Soudjani}
is a postdoctoral researcher at the Max Planck Institute for Software Systems, Germany.
His research interests are formal synthesis, abstraction, and verification of
complex dynamical systems with application in Cyber-Physical Systems, particularly
involving power and energy networks, smart grids, and transportation systems.
Sadegh received the B.Sc. degrees in electrical engineering and pure
mathematics, and the M.Sc. degree in control engineering from the University of Tehran,
Tehran, Iran, in 2007 and 2009, respectively.
He received the Ph.D. degree in Systems and Control in November 2014 from
the Delft Center for Systems and Control at the Delft University of Technology,
Delft, the Netherlands.
His PhD thesis received the best thesis award from the Dutch Institute of Systems and Control. 
Before joining Max Planck Institute, he was a postdoctoral researcher at the department of
Computer Science, University of Oxford, United Kingdom.
\end{IEEEbiography} 

\end{document}